\documentclass[11pt]{article}
\usepackage{extras}
\usepackage{tabularx,hyperref,amssymb,nicefrac,amsmath,multirow}
\usepackage{comment,amsmath,amssymb,amsfonts} 
\usepackage{epsfig} \usepackage{latexsym,nicefrac,bbm}
\usepackage{xspace}
\usepackage{color,fancybox,graphicx,url,subfigure,fullpage}
\usepackage{tabularx} 
\usepackage{hyperref} 

\usepackage{enumitem}
\usepackage{booktabs}

\newcommand{\nfrac}{\nicefrac}

\long\def\symbolfootnote[#1]#2{\begingroup
\def\thefootnote{\fnsymbol{footnote}}\footnote[#1]{#2}\endgroup}

\newcommand{\CG}{\mbox{${\mathcal G}$}}
\newcommand{\CF}{\mbox{${\mathcal F}$}}
\newcommand{\CP}{\mbox{${\mathcal P}$}}

\newcommand{\CM}{\mbox{${\mathcal M}$}}

\newcommand{\CE}{\mbox{${\mathcal E}$}}

\newcommand{\ppj}{\mbox{${p'_j}$}}
\newcommand{\cj}{\mbox{${\mathfrak{e}_j}$}}

\newcommand{\desire}{\mbox{\rm desire}}

\newcommand{\bpb}{\mbox{\rm bpb}}

\newcommand{\MM}{\mbox{\boldmath $M$}}
\newcommand{\one}{\mbox{\boldmath $1$}}
\newcommand{\zero}{\mbox{\boldmath $0$}}

\newcommand{\ra}{\rightarrow}
\newcommand{\R}{\mathbb{R}}
\newcommand{\Rplus}{\R_+}

\newcommand{\li}{{\lambda_i}}

\newcommand{\yy}{\mbox{\boldmath $y$}}

\newcommand{\CA}{\mbox{${\mathcal A}$}}

\newcommand{\ppp}{\mbox{\boldmath $p'$}}

\newcommand{\pp}{\mbox{\boldmath $p$}}

\newcommand{\pv}{\mbox{\boldmath $v$}}
\newcommand{\py}{\mbox{\boldmath $y$}}
\newcommand{\pq}{\mbox{\boldmath $q$}}
\newcommand{\pr}{\mbox{\boldmath $r$}}
\newcommand{\pa}{\mbox{\boldmath $a$}}

\newcommand{\ps}{\mbox{\boldmath $s$}}

\newcommand{\pgamma}{\mbox{\boldmath $\gamma$}}
\newcommand{\plambda}{\mbox{\boldmath $\lambda$}}

\newcommand{\xx}{\mbox{\boldmath $x$}}

\newcommand{\dl}{\mbox{${\delta}$}}
\newcommand{\gijk}{\mbox{${\gamma^i_{jk}}$}}
\newcommand{\uijk}{\mbox{${u^i_{jk}}$}}
\newcommand{\lijk}{\mbox{${l^i_{jk}}$}}
\newcommand{\qijk}{\mbox{${q^i_{jk}}$}}
\renewcommand{\ppj}{\mbox{${p'_j}$}}

\renewcommand{\cj}{\mbox{${c_j}$}}

\renewcommand{\CA}{\mbox{${\cal A}$}}
\renewcommand{\CP}{\mbox{${\cal P}$}}
\newcommand{\CY}{\mbox{${\cal Y}$}}
\renewcommand{\CE}{\mbox{${\cal E}$}}

\newcommand{\pc}{\mbox{\boldmath $c$}}
\newcommand{\wi}{\mbox{\boldmath $w^i$}}

\newcommand{\pE}{\mbox{\boldmath $\CE$}}

\renewcommand{\pv}{\mbox{\boldmath $v$}}
\renewcommand{\py}{\mbox{\boldmath $y$}}
\newcommand{\cc}{\mbox{\boldmath $c$}}
\renewcommand{\pq}{\mbox{\boldmath $q$}}
\renewcommand{\pr}{\mbox{\boldmath $r$}}

\newcommand{\pbeta}{\mbox{\boldmath $\beta$}}
\renewcommand{\pa}{\mbox{\boldmath $a$}}

\renewcommand{\ps}{\mbox{\boldmath $s$}}

\renewcommand{\one}{\mbox{\boldmath $1$}}
\renewcommand{\zero}{\mbox{\boldmath $0$}}

\newcommand{\zb}{\mbox{$z_\circ$}}
\newcommand{\zs}{\mbox{$z_*$}}
\newcommand{\pppb}{\mbox{\boldmath $p'_\circ$}}
\newcommand{\ppps}{\mbox{\boldmath $p'_*$}}
\newcommand{\qqb}{\mbox{\boldmath $q_\circ$}}
\newcommand{\qqs}{\mbox{\boldmath $q_*$}}
\newcommand{\rrb}{\mbox{\boldmath $r_\circ$}}
\newcommand{\rrs}{\mbox{\boldmath $r_*$}}
\newcommand{\yys}{\mbox{\boldmath $y_*$}}
\newcommand{\yyb}{\mbox{\boldmath $y_\circ$}}
\newcommand{\bbb}{\mbox{\boldmath ${\beta_\circ}$}}
\newcommand{\bbs}{\mbox{\boldmath ${\beta_*}$}}
\newcommand{\llb}{\mbox{\boldmath ${\lambda_\circ}$}}
\newcommand{\lls}{\mbox{\boldmath ${\lambda_*}$}}
\newcommand{\ggb}{\mbox{\boldmath ${\gamma_\circ}$}}
\newcommand{\ggs}{\mbox{\boldmath ${\gamma_*}$}}

\def\CA{{\mathcal A}}
\def\CP{{\mathcal P}}
\def\CC{{\mathcal C}}
\def\CG{{\mathcal G}}
\def\CM{{\mathcal M}}
\def\CP{{\mathcal P}}
\def\CF{{\mathcal F}}

\date{}
\title{On Computability of Equilibria in Markets with Production}
\author{\Large\em Jugal Garg\thanks{College of Computing, Georgia Institute of Technology, Atlanta, GA 30332-0280. Email:
{\sf jgarg@cc.gatech.edu}. \newline Supported by Algorithms and Randomness Center (ARC) Postdoctoral Fellowship.} \and {\Large\em
Vijay V.  Vazirani}\thanks{College of Computing, Georgia Institute of Technology, Atlanta, GA 30332--0280.  Email: {\sf
vazirani@cc.gatech.edu}. Research supported by NSF Grants CCF-0914732 and CCF-1216019, and a Guggenheim Fellowship.}}

\begin{document}
\maketitle

\thispagestyle{empty}
\begin{abstract}
Although production is an integral part of the Arrow-Debreu market model, most of the work in theoretical computer
science has so far concentrated on markets without production, i.e., the {\em exchange economy}. This paper takes a
significant step towards understanding computational aspects of markets with production.  

We first define the notion of separable, piecewise-linear concave (SPLC) production by analogy with SPLC utility functions.
We then obtain a linear complementarity problem (LCP) formulation that captures exactly the set of equilibria for
Arrow-Debreu markets with SPLC utilities and SPLC production, and we give a complementary pivot algorithm for finding an
equilibrium.  This settles a question asked by Eaves in 1975 \cite{eaves} of extending his complementary pivot algorithm to
markets with production.

Since this is a path-following algorithm, we obtain a proof of membership of this problem in PPAD, using Todd, 1976. We also
obtain an elementary proof of existence of equilibrium (i.e., without using a fixed point theorem), rationality, and oddness
of the number of equilibria. We further give a proof of PPAD-hardness for this problem and also for its restriction to
markets with linear utilities and SPLC production.  Experiments show that our algorithm runs fast on randomly chosen
examples, and unlike previous approaches, it does not suffer from issues of numerical instability. Additionally, it is
strongly polynomial when the number of goods or the number of agents and firms is constant. This extends the result of
Devanur and Kannan (2008) to markets with production.

Finally, we show that an LCP-based approach cannot be extended to PLC (non-separable) production, by constructing an example
which has only irrational equilibria.
\end{abstract}

\newpage
\setcounter{page}{1}

\section{Introduction} \label{sec.intro}
Among the most novel and significant additions to the theory of algorithms and computational complexity over the last decade
are deep insights into the computability of equilibria, both Nash and market. However, within  the study of market
equilibria, most of this work was concentrated on the exchange economy,
\cite{DPSV,JainAD,CMV,CPV,CSVY,Chen.plc,ChenTeng,Ye-Ex,DK08,VY,orlin,GMSV,vegh1,CPY,DM}\footnote{This list is by no means
complete.}, i.e., the Arrow-Debreu (AD) model \cite{AD} without
production; as described in Section \ref{sec:rel-wor}, the results obtained so far for markets with production are quite
rudimentary. Production is, of course, central to the Arrow-Debreu model and to most economies, and this represents a
crucial gap in the current theory.  The purpose of this paper is to address this gap. 

For the exchange economy, once the case of linear utilities was settled with polynomial time algorithms
\cite{DPSV,DV1,JMS,GK,JainAD,Ye-AD}, the next case was understanding the computability of equilibria in markets with
separable, piecewise-linear concave (SPLC) utilities. A series of remarkable works resolved this long-standing open problem
and showed it is PPAD-complete \cite{Chen.plc,ChenTeng,VY}, not only for Arrow-Debreu markets but also for Fisher
markets\footnote{Fisher market is a special case of exchange economy.} 
\cite{scarf}.  However, the proof of membership in PPAD, given in \cite{VY}, was indirect -- using the characterization of
PPAD via the class FIXP \cite{EY07}. 

\cite{GMSV} gave a direct proof of membership in PPAD by giving a path-following algorithm. This is a complementary pivot 
algorithm using Lemke's scheme \cite{lemke}, and it builds on a classic result of Eaves \cite{eaves} giving a similar
algorithm for the linear exchange market. 
Another importance of the result of \cite{GMSV} was that experimental results conducted on randomly generated instances 
showed that their algorithm is practical; in particular, it does not suffer from issues of numerical instability that are
inherent in previous approaches, e.g., \cite{scarf.fp, smale}.
In the face of PPAD-completeness of the problem, and the current status of the 
P = PPAD question, this is the best one can hope for.

For the case of 2-Nash, the problem of finding a Nash equilibrium in a 2-player bimatrix game, a path-following algorithm and
membership in PPAD are provided by the classic Lemke-Howson algorithm \cite{lh}; this is also a complementary pivot 
algorithm. We note that whereas several complementary pivot algorithms have since been given for Nash equilibrium
\cite{lemke,ET,HP,koller1996,stengel-econ},
for market equilibrium, \cite{GMSV} and ours are the only such algorithms since Eaves' work.

An example of Mas-Colell, having only irrational equilibria for Leontief utilities (stated in \cite{eaves}), shows that a
similar approach is not feasible for arbitrary piecewise-linear, concave utilities. The latter case lies in the class FIXP
and proving it FIXP-hard remains an open problem.

Our work on markets with production is inspired by this development and brings it on par with the current status of exchange
markets. On the one hand, this development made our task easier. On the other hand, unlike the case of exchange markets,
where a demarcation between rational and irrational equilibria was already well known before \cite{GMSV} embarked
on their work (see \cite{GMSV} for a detailed discussion of this issue),
for markets with production, no such results were known, making our task harder.

We first define the notion of separable, piecewise-linear concave (SPLC) production by analogy with SPLC utility functions. 
We then obtain a linear complementarity problem (LCP) formulation that captures exactly the set of equilibria for Arrow-Debreu
markets with SPLC utilities and SPLC production, and we further give a complementary pivot algorithm for finding an equilibrium.
This settles a question asked by Eaves in 1975 \cite{eaves} of extending his complementary pivot algorithm to markets with production.

Since this is a path-following algorithm, we obtain a proof of membership of this problem in PPAD, using Todd, 1976. 
We also obtain an elementary proof of existence of 
equilibrium (i.e., without using a fixed point theorem), rationality, and oddness of the number
of equilibria. We further give a proof of PPAD-hardness for this problem and also for its restriction to 
Arrow-Debreu markets with linear utilities and SPLC production.

Experiments show that our algorithm is practical. 
We note that market equilibrium algorithms, especially those involving production, are important in practice
\cite{Shoven,ellickson}.  Because of the PPAD-completeness of the problem, ours is among the best available avenues to
designing
equilibrium computing algorithms that can have an impact in practice. We further show that our algorithm is strongly polynomial 
when the number of goods or the number of agents and firms is constant. This extends the result of 
\cite{DK08} to markets with production.

Finally, we show that an LCP-based approach cannot be extended to PLC (non-separable) production, by constructing an example which 
has only irrational equilibria. We expect this case to be FIXP-complete, even if utilities are linear.

Besides being practical, complementary pivot algorithms have the additional advantage that they have provided deep insights into the problems studied in the past.
A case in point is the classic Lemke-Howson algorithm \cite{lh} for computing a Nash equilibrium of a 2-person bimatrix game,
where besides oddness of the number of equilibria, it yielded properties such as index, degree, and stability \cite{arndt,shapley,g2,AS}.
As stated above, we have already established that our problem has an odd number of equilibria. We expect our algorithm
to yield additional insights as well.

\subsection{Salient features}
Our result involves two main steps. The first is deriving an LCP whose solutions are exactly the set of equilibria
of our market with production. The second is ensuring that Lemke's scheme is guaranteed to converge to a solution. 

Lemke's scheme involves following the unique path that starts with the primary ray on the one-skeleton of the associated
polyhedron (see Appendix \ref{sec.LCP} for detailed explanation of these terms).
Such a path can end in two ways, either a solution to the LCP or a secondary ray. In the latter case, the scheme provides
no recourse and simply aborts without finding a solution. We show that the associated polyhedron of our LCP has no
secondary rays and therefore Lemke's scheme is guaranteed to give a solution. 

Several classes of LCPs have been identified for which Lemke's scheme converges to a solution \cite{murty,cottle}. However,
none of these classes captures our LCP, or the LCPs of Eaves \cite{eaves} for linear exchange markets or \cite{GMSV} for
SPLC exchange markets, even though they resort to a similar approach. In the progression of these three works, the LCPs have
become more involved and proving the lack of secondary rays has become increasingly harder. Clearly, this calls for further
work to understand the underlying structure in these LCPs.

Some new ideas were needed for deriving the LCP. The LCP has to capture: (i) optimal production plans for each firm,
(ii) optimal bundle for each agent, and (iii) market clearing conditions. 
Given prices, the optimal production plan of each firm can be obtained through a linear program (LP) using variables
capturing amount of raw and produced goods. Then, using complementary slackness and feasibility conditions of these LPs, we
obtain an LCP to capture the production. Next, an LCP for consumption and market clearing is sought using
variables capturing amount of goods consumed in order to merge it with the production LCP, however it turned out to be unlikely (see
Remark \ref{arem.amtvar} in Section \ref{asec.charNlcp} for details) and the correct way is to use variables capturing value of goods. The only
way out was to somehow convert amount variables in production LCP to value variables, which fortunately was doable.

The resulting LCP captures market equilibria, however it has non-equilibrium solutions as well and more importantly, it is
homogeneous -- the corresponding polyhedron forms a cone, with origin being the only vertex. Similar issues arise in \cite{eaves,GMSV}
too, and they deal with these by simply imposing a lower bound of $1$ on every price variable\footnote{This is without loss
of generality as market equilibrium prices are scale invariant \cite{AD}}. 
It turns out that such a trivial lower bound does not work for our LCP; prices where no firm can make positive profit are 
needed (see Section \ref{asec.nhlcp} for details). We recourse to the sufficiency condition {\em no production out of
nothing} \cite{AD,maxfield} for the existence of equilibria and show that such prices exist (using Farkas' lemma) and can be
obtained. After imposing such a lower bound, the resulting LCP exactly captures market equilibria (up to scaling). 

Next we show that our algorithm is strongly polynomial when either the number of goods or the number of agents and firms are
constant, extending the result of Devanur and Kannan \cite{DK08} to markets with production. 
For this, we decompose a constant dimensional space into polynomially many regions and show that every region can contain at
most one vertex traversed by our algorithm. 

Since our algorithm follows a complementary path, it together with Todd's result \cite{todd} on locally orienting such paths, proves
that the problem is in PPAD.
Since every LCP has a vertex solution
(if a solution exists) in the polyhedron associated with it, there is an equilibrium with rational prices. In the absence of secondary rays, all
but one of the equilibria get paired up through complementary paths; the remaining one with the primary ray. Therefore 
there are odd number of equilibria. 

Market equilibrium computation in exchange markets with SPLC utilities is known to be PPAD-hard \cite{Chen.plc,VY}.  We reduce such a
market to a market with linear utilities and SPLC production, thereby proving its hardness too.  This reduction is general
enough, in a sense that an exchange market with concave utilities can be reduced to an equivalent market with linear
utilities and concave production. Since linear is a special case of SPLC functions, we obtain PPAD-completeness for markets
with SPLC utilities and SPLC production. 

\subsection{Related work}\label{sec:rel-wor}
Jain and Varadarajan \cite{JV} studied the Arrow-Debreu markets with production, and gave a polynomial time algorithm for
production and utility functions coming from a subclass of CES (constant elasticity of substitution) functions; i.e.,
constant returns to scale (CRS) production.  They also gave a reduction from the exchange market with CES utilities to a
linear utilities market in which firms have CES production. Our reduction from the exchange market to a linear utilities
market with arbitrary production is inspired by their reduction but is more general.

We note that CRS production is relatively easy to deal with, since there is no positive profit to any firm at an
equilibrium. To the best of our knowledge, no computational work has been done for the original Arrow-Debreu market with decreasing
returns to scale production. 

\subsection{Organization of the paper}
In Section \ref{sec.model}, we define Arrow-Debreu market model and present an example of markets with linear utilities and
PLC production having only irrational equilibria. Markets with SPLC utilities and SPLC production, and the sufficiency
conditions for the existence of equilibrium are defined in Section \ref{asec.model}. Equilibrium characterization as
complementarity conditions is presented in Section \ref{asec.charNlcp}. In Section \ref{asec.nhlcp}, we derive an LCP
formulation which captures exactly the set of equilibria of markets with SPLC utilities and SPLC production. The algorithm
and the proof of convergence appears in Section \ref{asec.alg}.  Section \ref{asec.constant} presents a strongly polynomial
bound for our algorithm when either the number of goods or the number of agents plus firms is constant. For the hardness
result, we derive a general reduction from an exchange market to a market with production with linear utilities in Section
\ref{asec.exctoprod}.  Finally, experimental results are presented in Section \ref{asec.exp}, demonstrating that our
algorithm is practical. 

\section{The Arrow-Debreu Market Model}\label{sec.model}
The market model defined by Arrow and Debreu \cite{AD} consists of the following: A set $\CG$ of divisible goods, a set
$\CA$ of agents and a set $\CF$ of firms. Let $n$ denote the number of goods in the market. 

The production capabilities of firm $f$ is defined by a set of production possibility vectors (PVVs) $\CY^f$; in a vector
negative coordinates represent inputs and positive coordinates represents output. The set of input and output goods for each
firm is disjoint. Standard assumptions on set $\CY^f$ are (see \cite{AD}):

\begin{enumerate}
\item Set $\CY^f$ is closed and convex; convexity captures law of diminishing returns.
\item {\em Downward close} - additional raw goods do not decrease the production.
\item {\em No production out of nothing} - firms together can not produce something out of nothing, i.e., $\oplus_{f \in \CF} \CY^f
\cap \mathbb R^n_+ = \zero$.  
\end{enumerate}

The goal of a firm is to produce as per a profit maximizing (optimal) schedule.  Firms are owned by agents: $\theta^i_f$ is
the profit share of agent $i$ in firm $f$ such that $\forall f\in \CF,\ \sum_{i\in\CA} \theta^i_f=1$.

Each agent $i$ comes with an initial endowment of goods; $w^i_j$ is amount of good $j$ with agent $i$.
The preference of an agent $i$ over bundles of goods is captured by a
non-negative, non-decreasing and concave utility function $U^i:\R_+^n\rightarrow \R_+$. $U^i$ can be assumed to be non-decreasing is due to free
disposal property, and concavity captures the law of diminishing marginal returns. 
Each agent wants to buy a (optimal) bundle of goods that maximizes her utility to the extent allowed by her earned money --
from initial endowment and profit shares in the firms.  Without loss of generality, we assume that total initial endowment
of every good is 1, i.e, $\sum_{i\in\CA} w^i_j = 1, \forall j\in \CG$\footnote{This is like redefining the unit of goods by
appropriately scaling utility and production parameters.}.

Given prices of goods, if there is an assignment of optimal production schedule to each firm and optimal affordable bundle
to each agent so that there is neither deficiency nor surplus of any good, then such prices are called {\em market
clearing} or {\em market equilibrium} prices. The market equilibrium problem is to find such prices when they exist.
In a celebrated result, Arrow and Debreu \cite{AD} proved that market equilibrium always exists under some mild conditions,
however the proof is non-constructive and uses heavy machinery of Kakutani fixed point theorem. 

Note that operating point of a firm, at any given prices, is on the boundary of $\CY^f$, which can be defined by a concave
function/correspondence. To work under finite precision model concave is generally approximated with piecewise-linear concave. 

A well studied restriction of Arrow-Debreu model is {\em exchange economy}, i.e., markets without production firms.

\subsection{PLC production and irrationality}\label{asec.eg}
In this section we demonstrate an example of a market with (non-separable) PLC production and the simplest utility
functions, namely {\em linear}, having only irrational equilibrium prices and allocations.

Consider a market with three goods, three agents and one firm. The initial endowments of agents are $w^1=w^2=w^3=(1,1,0)$. Each utility
function has one linear segment; $U_1=x^1_1$, $U_2=x^2_2$ and $U_3=x^3_3$. The firm is owned by agent $3$, i.e.,
$\theta^3_1=1$. It has exactly one production segment without any upper
limit on the raw material used, and needs two units of good $1$ and a unit of good $2$ to produce a unit of good
$3$. This is a Leontief (PLC) production function where the quantities of raw goods are needed in a fixed proportion. Let
$r_j$'s and $s_j$'s respectively be the amount of goods used and produced by the firm on its only segment, then they should
satisfy: $2 \cdot s_3\le r_1$ and $s_3\le r_2$. 

Let $\pp=(p_1,p_2,p_3)$ denote the price of goods. Note that at an equilibrium of this market all prices must be positive,
otherwise the demand of zero priced goods will be infinite. Since equilibrium prices are scale invariant \cite{AD}, we set
$p_1=1$. The firm will produce at equilibrium due to positive demand of good $3$ from agent $3$, however its profit will
be zero otherwise it will want to produce infinite amount. Hence we have $p_3=2+p_2$. From the market clearing conditions,
we get $p_2^2+2p_2-2=0$. Thus the only equilibrium prices of this market are $p_1=1, p_2=\sqrt{3}-1$ and
$p_3=\nfrac{(1+\sqrt{3})}{2}$. At equilibrium the allocation and production variables are: $x^1_1=\sqrt{3}$,
$x^2_2=\nfrac{\sqrt{3}}{(\sqrt{3}-1)}$ and $x^3_3=s_3=\nfrac{r_1}{2}=r_2=\nfrac{\sqrt{3}}{(\sqrt{3}+1)}$.

This rules out the possibility of linear complementarity problem (LCP) formulation for PLC production, because LCPs always
have a rational solution (if one exists) given that all input parameters are rational numbers. The next logical step is to
consider separable PLC (SPLC) production instead, together with SPLC utilities, as PLC utility functions are already known
to have irrationality \cite{eaves1}.

\section{Markets with SPLC Utility \& SPLC Production}\label{asec.model}
In this section, we define parameters representing separable piecewise-linear concave (SPLC) utility and SPLC
production functions. All the parameters are assumed to be rational numbers.

For each pair of agent $i$ and good $j$ we are specified a non-decreasing, piecewise-linear and concave (PLC) function $U^i_j:
\Rplus \ra \Rplus$, which gives the utility that $i$ derives as a function of the amount of good $j$ that she receives. Her
overall utility, $U^i(\xx)$, for a bundle $\xx=(x_1,\ldots,x_n)$ of goods is additively separable over the goods, i.e.,
$U^i(\xx) = \sum_{j\in\CG} U^i_j(x_j)$. The number of segments in function $U^i_j$ is denoted by $|U^i_j|$, and the $k^{th}$
segment of $U^i_j$ by $(i, j, k)$. The slope of a segment specifies the rate at which the agent derives
utility per unit of additional good received. Suppose segment $(i, j, k)$ has domain $[a, b] \subseteq \Rplus$, and slope
$c$. Then, we define $u^i_{jk} = c$ and $l^i_{jk} = b-a$. The length of the last segment is infinity. Since $U^i_j$ is concave,
$u^i_{j(k-1)} > u^i_{jk},\ \forall k\ge 2$.

In this paper we consider the case where every firm produces a good using a set of goods as raw material and the production
function is additively separable over goods.  For simplicity, we assume that each firm produces exactly one
good\footnote{This is without loss of generality since production is separable. For a firm producing multiple goods we can
create as many firms as number of produced goods with agent's shares being duplicated.}. Let firm $f$ produces good $j_f$.

For each pair of firm $f$ and good $j$ we are specified a production function $P^f_{j}:\Rplus \ra \Rplus$ which is non-negative,
non-decreasing, PLC, and defines $f$'s ability to produce good $j_f$ as a function of the amount of good
$j$. The overall production of the firm $f$ from a bundle $\xx=(x_1,\ldots,x_n)$ of goods is: $P^f(\xx)=\sum_{j \in \CG}
P^{f}_{j}(x_{j})$. The number of segments in $P^f_{j}$ is denoted by $|P^f_{j}|$, and the $k^{th}$ segment by $(f,j,k)$.
The slope of a segment specifies the rate at which good $j_f$ can be produced from a unit of additional good $j$. 
Suppose segment $(f,j,k)$ has domain $[a,b]\subseteq \Rplus$, and slope $c$, then we define $\alpha^f_{jk} = c$ and $o^f_{jk}=b-a$. 
This implies that on segment $(f,j,k)$, firm $f$ can produce $\alpha^f_{jk}$ amount of good $j_f$ from a unit amount of good
$j$ and at this rate it can use up to $o^f_{jk}$ units of good $j$.  The length of the last segment is infinity. 
Since $P^{f}_{j}$ is concave, $\alpha^f_{j(k-1)} > \alpha^f_{jk}\ \forall k\ge 2$. 

Given prices $\pp = (p_1, \ldots, p_n)$ for the goods, each firm operates on a production
schedule that maximizes its profit $-$ money earned from the production minus the money spent on the raw material. 
Let $\CE^f$ denote the
profit of firm $j$. Agent $i$ earns $\sum_{j \in \CG} w^i_{j}p_j$ from the initial endowment and $\sum_{f \in \CF}
\theta^i_{f} \CE^f$ from the profit shares in firms, and buys a bundle of goods that maximizes her utility.  Prices $\pp$
gives an equilibrium if market clears when each firm produces at an optimal plan and each agent buys an optimal bundle.
We will denote this market, with SPLC production and SPLC utilities, by $\CM$.

\subsection{Sufficiency conditions} \label{asec.strong}
In general there may not exist market equilibrium prices; in fact, for the special case of exchange market with SPLC
utilities, it is NP-hard to determine if they exist \cite{VY}. However, an equilibrium is guaranteed to exist under certain
sufficient conditions.  The nine conditions given by Maxfield \cite{maxfield} are the weakest known sufficiency conditions for
the existence. Out of these, seven are trivially satisfied by our model. The remaining two are described next.

One is the condition (3) of Section \ref{sec.model}, which translates to the following for the SPLC production:
It says that firms together can not produce something out of nothing
\cite{maxfield,AD}.  For example, suppose firm $1$ can produce a unit amount of good $1$ from a unit of good $2$ and from
this unit amount of good $1$ suppose firm $2$ can produce $2$ units of good $2$. At the end of such a production we have no
decrease in any good, and good $2$ is increased by one unit. In other words good $2$ can be produced from nothing, which is
unnatural. 
In addition we disallow {\em vacuous productions} as well\footnote{This was assumed by Arrow-Debreu, but not by Maxfield. We
enforce this to avoid unnecessary complications in the proofs later. Dropping this assumption will introduce degeneracy,
which can be handled.}, and assume that in every production cycle quantity of at least one good reduces.

In case of SPLC production, since the production functions are concave, the returns
are decreasing. Hence, checking these conditions using the first segments of $P^f_{j}$'s suffices for all the combinations
of production possibility vectors.
Firm $f$ needs at least $\nfrac{1}{\alpha^f_{j1}}$ units of good $j$ to produce one unit of good
$j_f$. Let $G_\CF(\CM)$ be a weighted directed graph, where goods are nodes and the weight of an edge from $j$ to $j'$ is
$\max_{f, j_f=j'} \nfrac{1}{\alpha^f_{j1}}$. 

\begin{definition}[No production out of nothing \cite{AD,maxfield}]\label{adef.np}
We say that market $\CM$ satisfies no production out of nothing if weights of edges in every cycle of $G_\CF(\CM)$ multiply
to strictly less than one.
\end{definition}

For the second condition, we say that {\em agent $i$ is non-satiated by good $j$} if the last
segment of $U_j^i$ has positive slope, i.e., $u^{i}_{j|U^i_j|}>0$. Similarly, firm $f$ is non-satiated by good $j$ if the last segment
of $P^f_{j}$ has positive slope, i.e., $\alpha^f_{j|P^f_{j}|}>0$.

\begin{definition}[Strong connectivity]\label{adef.sc} 
Construct a directed graph $G(\CM)$ whose nodes correspond to agents and firms of
market $\CM$ and there is an edge from node $a$ to node $b$ if there is a good possessed/produced by node $a$ for which node $b$
is non-satiated.  Market $\CM$ satisfies strong connectivity if this graph has a strongly connected component containing all
the agent nodes.  
\end{definition}

This condition also makes sure that every agent has a positive amount of at least one good in her initial endowment.
Henceforth we will assume that market $\CM$ satisfies {\em no production out of nothing} and {\em strong connectivity} conditions.

\section{Equilibrium Characterization and LCP Formulation}\label{asec.charNlcp}
For equilibrium characterization, we need to capture $(i)$ optimal production plans for each firm, $(ii)$ optimal bundles
for each agent, and $(iii)$ market clearing conditions. 
\medskip

\noindent{\bf Optimal production.}
Recall that on segment $(f,j,k)$, $\alpha^f_{jk}$ units of good $j_f$ can be produced using a unit of good $j$. Given prices
$\pp$, the optimal production plan of firm $f$ is given by the following linear program (LP), where $x^f_{jk}$ denote the
amount of raw good $j$ used by firm $f$ on $(f,j,k)$: 

\begin{eqnarray}\label{alp}
\hspace{1cm}\mbox{maximize} \ \ \ \ \sum_{j,k} x^f_{jk}(\alpha^f_{jk} p_{j_f} - p_{j})\ \ \  \mbox{ subject to } \ \ \  
\ \ 0 \le x^f_{jk} \le o^f_{jk},\ \ \forall(j,k) 
\end{eqnarray}

Note that since $\alpha^f_{jk} > \alpha^f_{j(k+1)}, \forall k$, an optimal solution of this LP will have $x^f_{jk}=o^f_{jk}$
whenever $x^f_{j(k+1)}>0$. This is required for an optimal production plan as $x^f_{jk}$'s have to be allocated in order. 
Let $\beta^f_{jk}$ be the dual variable corresponding to inequality $x^f_{jk}\le o^f_{jk}$. From the
optimality conditions,
we get the following linear constraints and complementarity conditions (We will refer to these as follows: the
equation number will refer to the constraint and the equation number with a prime will refer to the complementarity
condition, e.g., (\ref{aeq01}) refers to the first constraint below and (\ref{aeq01}') refers to the corresponding
complementarity condition.). All variables introduced will have a non-negativity constraint; for the sake of brevity, we
will not write them explicitly.

\begin{equation}\label{aeq01}
\forall(j,k): \ \ \ \alpha^f_{jk}p_{j_f} - p_{j} \le \beta^f_{jk}\ \ \ \ \ \  \mbox{and} \ \ \ \ \ x^f_{jk}(\alpha^f_{jk}p_{j_f} - p_{j} -
\beta^f_{jk}) = 0
\end{equation}

\begin{equation}\label{aeq02}
\forall(j,k): \ \ \ x^f_{jk} \le o^f_{jk}\ \ \ \ \ \ \ \ \mbox{and} \ \ \ \ \ \ \beta^f_{jk}(x^f_{jk} - o^f_{j'k}) = 0 
\end{equation}

Note that (\ref{aeq01}) and \ref{aeq02}) are equivalent to the above LP due to strong duality. Combining these for all the
firms gives us a linear complementarity problem (LCP) formulation that capture optimal production. 
The amount produced on segment $(f,j,k)$ is \[y^f_{jk}=\alpha^f_{jk} x^f_{jk}.\]

Next we need to characterize and derive an LCP to capture optimal bundles of each agent, and market clearing conditions.
Suppose $x^i_{jk}$ denote the amount of good $j$ obtained by agent $i$ on segment $(i,j,k)$, then $i$ spends $\sum_{j,k}
x^i_{jk}p_j$, which is a quadratic term. It turned out to be unlikely to capture this through an LCP in amount variables;
see Remark \ref{arem.amtvar} for details. Therefore we need to use $q^i_{jk}$ in place of
$x^i_{jk}p_j$, representing money spent by
agent $i$ on segment $(i,j,k)$. Further, we need same type of variables in production LCP to tie everything in goods side
market clearing condition.

\begin{remark}\label{arem.amtvar}
Suppose there is an LCP in amount variables $\xx$ and price variables $\pp$ for the linear exchange market, where $x_{ij}$
denotes the amount of good $j$ allocated to agent $i$. 
The idea is to write optimal production plans for firms with Leontief production as linear complementarity constraints
in $(\xx,\pp)$ and then incorporate them with the LCP for the linear exchange market. For example, the
optimal production plan for a firm which produces $1$ unit of good $3$ using $b$ units of good $1$ and $c$ units of good $2$
can be written as follows, where $r_1$ and $r_2$ respectively captures the amount of goods $1$ and $2$ used as raw
materials and $s_3$ captures the amount of good $3$ produced:

\begin{center}
$p_3 \le bp_1 + cp_2$, $\ \ \ \ $   $s_3 \ge 0$, $\ \ \ \ $  $s_3(p_3-bp_1-cp_2)=0$\\
$r_1 = b s_3$, $\ \ \ \ $  $r_2 = c s_3$
\end{center}

Since production of this firm is constant returns to scale, it can not make positive profit at an equilibrium. Further,
it produces if and only if the profit is zero.
The first set of constraints capture the same. The second set of equalities capture the amount of raw material needed to
produce $s_3$ units of good $3$.
Similarly we can write linear complementarity constraints for all the firms. They together with the LCP for the linear
exchange model, with variables $s_i$'s and $r_i$'s added at appropriate places, give us an LCP for markets with linear
utilities and Leontief production, contradicting its irrationality (see Section \ref{asec.eg}). 
\end{remark}

The only recourse is to convert amount variables in production LCP to money variables. It turns out to be doable,
though not immediately clear, using complementarity and change of variables. We multiply both the equations in (\ref{aeq02}) by $p_{j}$
and replace the expression $x^f_{jk}p_j$ by $r^f_{jk}$ denoting the money spent on raw material $j$ on segment $(f, j,k)$.  Assuming
that $p_j >0, \forall j\in\CG$ at equilibrium, conditions (\ref{aeq01}) and (\ref{aeq02}) can be replaced with the following.

{\small
\begin{equation}\label{aeq1}
\forall(f,j,k): \ \ \alpha^f_{jk}p_{j_f} - p_{j} \le \beta^f_{jk}\ \  \ \ \ \mbox{and} \ \ \ \ \ \
r^f_{jk}(\alpha^f_{jk}p_{j_f} - p_{j} - \beta^f_{jk}) = 0
\end{equation}
\begin{equation}\label{aeq2}
\forall(f,j,k): \ \ r^f_{jk} \le o^f_{jk}p_{j}\ \ \ \  \ \ \mbox{and} \ \ \ \ \ \ \beta^f_{jk}(r^f_{jk} - o^f_{jk}p_{j}) = 0 
\end{equation}
}

To capture produced amount let $s^f_{jk}$ denote the revenue of firm $f$ on segment $(f,j,k)$, namely $y^f_{jk}p_{j_f}$. Directly
replacing $y^f_{jk}$ with $\nfrac{s^f_{jk}}{p_{j_f}}$ and $x^f_{jk}$ with $\nfrac{r^f_{jk}}{p_j}$ in $y^f_{jk}=\alpha^f_{jk}x^f_{jk}$ 
will give quadratic equality.  Instead, we observe that $\beta^f_{jk}$ captures the profit per unit of raw material on segment
$(f,j,k)$ when $r^f_{jk}>0$.  Further, if $\beta^f_{jk}>0$, then the firm utilizes segment $(f,j,k)$ completely. 
Putting these together we include the following equalities, where $\CE^f$ captures the profit of firm $f$.

{\small
\begin{equation}\label{aeq3}
\forall (f,j,k): \ \ s^f_{jk}=r^f_{jk} + o^f_{jk}\beta^f_{jk}\ \ \ \   \mbox{ and } \ \ \ \ \forall f:\ \  \CE^f = \sum_{j,k}
o^f_{jk}\beta^f_{jk}
\end{equation}
}

Clearly, according to $f$'s optimal plan it produces the positive profit segments completely, and does 
not produce the segments giving negative profit at all. It is indifferent between producing on zero profit segments. 

\begin{remark}
We note that, even though one can write a linear program to compute optimal production plan in case of (non-separable) PLC
production as well, there can not exists an LCP formulation to capture equilibria for markets with PLC production. This
follows from the example of such a market in Section \ref{asec.eg} with only irrational equilibrium prices and allocations.  
\end{remark}

\noindent{\bf Market clearing.}
The market clearing constraints are easier now. 
Let $q_{jk}^i$ be a variable that denotes the amount of money spent by agent $i$ for buying good $j$ corresponding
to segment $(i, j, k)$. The following constraints capture the market clearing, where $\CF(j)$ denote the set of firms
producing good $j$, and $\li$ is related to optimal bundle of agent $i$ and is defined below; 
we have included the corresponding complementarity conditions in order to obtain an LCP in the standard form. 

{\small
\begin{equation} \label{aeq4} 
\begin{array}{c}
\forall j \in \CG:  \ \ \  \displaystyle\sum_{i,k} {q_{jk}^i} + \sum_{f,k} r^f_{jk} \leq
p_j + \sum_{f\in \CF(j),j',k} s^f_{j'k} \ \ \  \mbox{and} \ \ \ p_j(\sum_{i,k} {q_{jk}^i} + \sum_{f,k} r^f_{jk} - 
p_j - \sum_{f\in \CF(j),j',k} s^f_{j'k}) = 0 
\end{array}
\end{equation}

\begin{equation} \label{aeq5} 
\begin{array}{c}
\forall i \in \CA:  \ \ \ \displaystyle\sum_{j} {w_j^i p_j} + \sum_f{\theta^i_{f}\CE^f} \leq \sum_{j,k} {q_{jk}^i}  \ \ \  \mbox{and}  \
\ \ \li(\sum_{j} w_j^i p_j + \sum_f \theta^i_{f}\CE^f - \sum_{j,k} {q_{jk}^i}) = 0 
\end{array}
\end{equation}
}

\noindent{\bf Optimal bundle.}
Given prices $\pp$, earnings of each agent is fixed, i.e., for agent $i$ it is $\sum_j w^i_j p_j+\sum_f \theta^i_f \CE^f$.
Therefore, she will try to spend her money where utility per unit of money, also called bang-per-buck, is maximum; on segment $(i,j,k)$ it is
$\bpb^i_{jk}=\nfrac{u^i_{jk}}{p_j}$. We take $\nfrac{0}{0}$ as $0$.
Hence, optimal bundle of agent $i$ can be computed as follows: sort all her segments by decreasing bang per buck and partition them by
equality, i.e., each equivalence class will consist of all segments having equal bang-per-buck. Let the classes be: $Q_1, Q_2, \ldots$.  At
prices $\pp$, the segments in $Q_l$ make $i$ strictly happier than those in $Q_{l+1}, Q_{l+2}, \ldots$.
Hence, she would start buying partitions in order, until all her money is
exhausted. Suppose she exhausts all her money at $k_i^{th}$ partition. The segments in partitions 1 to $k_{i-1}$ will be
called {\em forced}, those in partition $k_i$ will be called {\em flexible} and those in partitions $k_{i+1}$ and higher
will be called {\em undesirable}. Indeed, every optimal bundle is obtained in this manner: it must fully allocate all
segments in the forced partitions; the money left over after this allocation is spent on segments in the flexible partition
in any manner, since all these segments have equal bang per buck; and no allocation is made corresponding to segments in
undesirable partitions.

Introduce a variable $\lambda_i$ that captures inverse of the bang-per-buck of the flexible partition. 
In order to capture segments of forced partition, variable $\gamma^i_{jk}$ is introduced so that if $(i,j,k)$ is forced, then
$\nfrac{1}{\lambda_i} = \nfrac{u^i_{jk}}{(p_j+\gamma^i_{jk})}$; supplement prices. 
The following constraints from \cite{GMSV} ensure optimal bundle to each agents.

{\small
\begin{equation} \label{aeq6} 
\forall (i,j,k):  \ \ \ \  u^i_{jk} \li \leq p_j + \gijk \ \ \ \  \mbox{and}  \ \ \ \ q^i_{jk}(u^i_{jk}\li -  p_j - \gijk) =
0 
\end{equation}

\begin{equation} \label{aeq7} 
\forall (i,j,k):  \ \ \ \  q^i_{jk} \leq l^i_{jk} p_j  \ \ \ \  \mbox{and}  \ \ \ \  \gijk ( q^i_{jk} - l^i_{jk} p_j) = 0
\end{equation}
}

Let us denote the LCP defined by the sets of constraints and complementarity conditions given in (\ref{aeq1}) through (\ref{aeq7}),
together with non-negativity on all variables, as {\bf AD-LCP}.
 
\begin{lemma} \label{alem.opt1} 
Any equilibrium of market $\CM$ yields a solution to AD-LCP.  
\end{lemma} 
\begin{proof}
Let $\pp$ be a market equilibrium prices.
Let $x^f_{jk}$'s be the bundle of raw material used by firms at the equilibrium. For segment $(f,j,k)$, set variables $r^f_{jk}=x^f_{jk}p_j$,
$\beta^f_{jk}=\alpha^f_{jk} p_{j_f}-p_j$, and $s^f_{jk}$ and $\CE^f$ accordingly. Since $x^f_{jk}$'s and $\beta^f_{jk}$'s satisfies
(\ref{aeq01}) and (\ref{aeq02}), it follows that for the set values conditions (\ref{aeq1}), (\ref{aeq1}'), (\ref{aeq2}) and
(\ref{aeq2}') are satisfied. Set $s^f_{jk}$ and $\CE^f$ as per (\ref{aeq3}) and (\ref{aeq3}') respectively; they are revenue and profit
at optimal
production. 

Let $x^i_{jk}$'s be the bundle of goods obtained by the agents at the equilibrium. Set $q^i_{jk}=x^i_{jk}p_j$. Agent $i$ earns 
$\sum_j w^i_jp_j + \sum_f \theta^i_f \CE^f$ and spends the same amount. Therefore (\ref{aeq4}) holds with equality, consequently
(\ref{aeq4}') is also satisfied. Similarly due to market clearing from good side we get that (\ref{aeq5}) too holds with equality.

For agent $i$ set variable $\lambda_i$ to the inverse of the bang-per-buck of her last bought segment.
substitute for the variables $\gijk$ as follows: if segment $(i, j, k)$ is flexible or
undesirable, set it to zero, and if it is forced, set it so that the following equality is satisfied \[ {1 \over {\li}} =
{u^i_{jk} \over {p_j + \gijk}} .\] Clearly, all the $\gijk$'s satisfy non-negativity. Now, it is easy to verify that in each
of the three cases -- that the segment $(i, j, k)$ is forced, flexible or undesirable -- the constraints (\ref{aeq6}) and
(\ref{aeq7}), and complementarity conditions (\ref{aeq6}') and (\ref{aeq7}') are all satisfied.
\end{proof}

Since LCPs always have rational solutions (if one exists), next corollary follows from Lemma \ref{alem.opt1}.

\begin{corollary}\label{acor.rational} In a market with SPLC utilities and SPLC production functions, equilibrium prices and
allocations are rational up to scaling.  \end{corollary}

\begin{corollary}\label{acor.splcnpc} Checking existence of an equilibrium in AD market with SPLC utilities and SPLC
production is NP-complete.  
\end{corollary} 
\begin{proof} Since exchange market with SPLC utilities is a special case, the NP-hardness follows from the result of \cite{VY}.
Given prices it can be checked in polynomial time if corresponding optimal production plan and optimal bundles yield market
clearing using their characterization, since all the values are rational (Corollary \ref{acor.rational}). Thus containment
in NP follows. 
\end{proof}

\noindent AD-LCP suffers from two shortcomings. First, since the rhs vector of the constraints, denoted by $\pq$ in Appendix 
\ref{sec.LCP}, is zero (homogeneous system), the polyhedron is highly degenerate -- in fact, it is a cone with its vertex at
the origin. Second, AD-LCP may admit solutions that do not correspond to equilibria; suppose there is a subset $\CG'
\subset \CG$ of goods, for which every firm producing a good in $\CG \setminus \CG'$ are satiated. Let $\overline{\CG'}=\CG
\setminus \CG'$. Suppose, for firm $f$, $j_f \in \overline{\CG'}$ can produce $d$ amount from a good $j' \in \CG'$, i.e., $d=\sum_{k\le
|P^f_{j}|} \alpha^f_{jk} o^f_{jk}$, then remove $P^f_{j}$ and change endowments of agents as $w^i_{j_f} =w^i_{j_f}
+\theta^i_f d$. Do this for every pair of $j_f \in \overline{\CG'}$ and $j \in \CG'$. Now, find an equilibrium for the market
consisting of agents of $\CA$, goods of $\overline{\CG'}$ and firms producing $\overline{\CG'}$.  Set the corresponding
variables in accordance with this equilibrium. For each good $j \in \CG'$, set $p_j = 0$ and for each segment $(i, j, k)$ of
this good corresponding to agents, set $q^i_{jk} = 0$ and $\gijk = \uijk \li$. Further, each segment $(f,j',k)$, where $j_f=j$,
set $r^f_{j'k}=\beta^f_{j'k}=0$ and accordingly $s^f_{j'k}$ and $\CE^f$. For a firm $f$ with $j_f \in
\overline{\CG'}$ and good $j \in \CG'$ if $\alpha^f_{jk}>0$, then set $\beta^f_{jk}=\alpha^f_{jk}p_j$ and accordingly
$s^f_{jk}$ and $\CE^f$. One can verify that this is a solution to AD-LCP, but may not be a market equilibrium.

The next section circumvents both these shortcomings by constructing a non-homogeneous LCP, where the rhs vector is non-zero
and has negative entries.

\section{Non-Homogeneous LCP}\label{asec.nhlcp} 
In this section first we ensure that prices are positive at every equilibrium, and then dehomogenize AD-LCP by imposing appropriate lower
bounds on price variables and show that it exactly captures the market equilibria. 

Recall that the lengths $o^f_{j|P^f_{j}|}$ and $l^i_{j|U^i_j|}$ of last segments, for all pairs of $(f,j)$'s and $(i,j)$'s, are
infinite.  Using the sufficiency conditions of Section \ref{asec.strong}, next we calculate their values that are never
reached at an equilibrium. Since no firm operates at loss at an equilibrium, we get the following:  

\begin{claim}
\label{acl.noprodcycle} Given prices $\pp \neq 0$, consider a graph where firms are nodes and there is an edge from $j_f$ to
$j$ if the profit on segment $(f,j,1)$ is non-negative. This graph is acyclic.  \end{claim} 
\begin{proof} If there is a cycle, then non-negative profit on each of those edges 
implies that the weights on the same cycle in $G_\CF(\CM)$ of {\em no production out of nothing} condition 
will multiply to at least one, contradicting the same condition.
\end{proof}

Claim \ref{acl.noprodcycle} implies that there cannot be a cycle of productions at any equilibrium, however there may be
chains. Total endowment of each good brought by agents is one. Let $\alpha_{max}=\max_{(f,j,k)}
\alpha^f_{jk}$ and $\alpha_{min}=\min_{(f,j,k),\alpha^f_{jk}\neq0} \alpha^f_{jk}$. It is easy to see that the maximum
production a firm can do at the end of a chain is $L=n^n(\alpha_{max}+1)^n$.  Therefore, we set
$o^f_{j|P^f_{j}|}=\nfrac{L}{\max\{1,\alpha_{min}\}}, \forall (f,j)$ and $l^i_{j|U^i_j|}=L+1,\ \forall(i,j)$, which are safe limits. 

For a good $j$, define $\desire(j)$ to be the total amount represented by its non-zero utility segments, {\em i.e.,}
$\desire(j)=\sum_{(i,k): u^i_{jk}>0} l^i_{jk}$. 

\begin{lemma}\label{ale.desire}
If $\desire(j)>1,\ \forall j \in \CG$ then $\pp>0$ in every equilibrium.
\end{lemma}
\begin{proof}
Let $\pp$ be an equilibrium prices. If no firm is producing at these prices, then agents' demand for a good with zero price will be
more than one, hence the lemma holds. Suppose firms are producing, 
then due to Claim \ref{acl.noprodcycle} there can be only paths of production among goods; last good, say $g$, on this path is not
getting produced. Further, if the price of any good is zero on this path, the price of good $g$ has to be zero, since a zero
priced good can be produced using zero priced goods only (or else there will be losses). 
In that case, demand of good will be more than supply, a contradiction.
\end{proof}

Henceforth we assume that $\desire(j)>1, \ \forall j \in \CG$, and call this condition {\em enough demand}\footnote{A stronger
condition may be derived to ensure non-zero prices at equilibrium, however we stick to this one 
for simplicity.}. In that case since the
equilibrium prices are positive, and they are known to be scale invariant, we can lower bound them with positive numbers. This will
lead to non-zero rhs in the LCP. We will want that negative rhs appears only in the agent side market clearing condition (\ref{aeq5}).
This is needed to ensure that all the equilibrium conditions except market clearing are satisfied on the path followed by the
algorithm, which is crucial to prove convergence of the algorithm.

Suppose, we lower bound $p_j$ by a positive number $c_j$. We do this by replacing $p_j$ with $p'_j+c_j$ in AD-LCP.
Note that, rhs of (\ref{aeq5}) will be surely negative, nothing can be said about (\ref{aeq1}) and the rest will be positive.
Now to keep the rhs of (\ref{aeq1}) non-negative, we need $c_j - \alpha^f_{jk} c_{j_f} \ge 0, \ \forall (f,j,k)$, i.e., no positive
profit on segment $(f,j,k)$ at prices $\pc$. We need a feasible point of the following polyhedron to compute such a vector $\pc$, 

\begin{equation}\label{aeq.lblp}
\begin{array}{ll}
\forall (f,j,k):& \alpha^f_{jk} c_{j_f} \le c_j \\
\forall j:&  c_j \ge 1
\end{array}
\end{equation}

Since the first condition of (\ref{aeq.lblp}) is homogeneous, setting all $c_j$'s to zeros is a solution, hence the second condition. 
Denote the polyhedron of (\ref{aeq.lblp}) by $\CC$. 

\begin{lemma}\label{ale.lb}
Polyhedron $\CC$ is non-empty and has a non-empty interior. 
\end{lemma}
\begin{proof}
Taking logarithms, the first condition of (\ref{aeq.lblp}) transforms to $\log(c_{j_f})-\log(c_j) \le -\log\alpha^f_{jk}$.
and the second condition to $\log(c_j)\ge 0$. Rename $\log(c_j)$ by $e_j$; 
this gives a system analogues to $Ax\le b$. By Farkas' lemma this does not have a solution if and only if there is a $y \ge0,\
y^TA=0,\ y^Tb=-1$ \cite{BSS}. It is easy to check that for our system of equations, existence of such a $y$ implies a cycle of weight
at least zero in the graph between goods, where there is an edge from $j$ to $j_f$ with weight $\alpha^f_{jk}$. This
contradicts {\em no production out of nothing} assumption. Further, this condition also implies that there is no cycle with
$\log\alpha^f_{jk}$'s adding to zero, hence $\CC$ has a non-empty interior.
\end{proof}

Take a vector $\cc$ from the interior of $\CC$; such a vector exists by Lemma \ref{ale.lb}.
Replace $p_j$ with $\ppj+c_j$ in AD-LCP, and the resulting LCP, call it {\bf NHAD-LCP}, is as follows. There are non-negativity
constraints on all the variables, however for brevity we omit them. 

{\small
\begin{equation}\label{aeq.01} 
\begin{array}{c} 
\forall (f,j,k):  \ \ \ \ \alpha^f_{jk}p'_{j_f}-\ppj -\beta^f_{jk} \le
c_{j}-\alpha^f_{jk}c_{j_f} \ \  \ \ \mbox{and}\ \ \ \ 
r^f_{jk}\left(\alpha^f_{jk}(p'_{j_f}+c_{j_f}) - (\ppj +c_j) 
-\beta^f_{jk}\right) = 0 
\end{array}
\end{equation}

\begin{equation}\label{aeq.02} 
\begin{array}{c} 
\forall (f,j,k): \ \ \ \ r^f_{jk} - o^f_{jk} \ppj\le o^f_{jk}c_{j} \ \ \ \ \mbox{and}\
\ \ \ \beta^f_{jk}(r^f_{jk} - o^f_{jk}(p'_{j}+c_{j}))=0 
\end{array}
\end{equation}

\begin{equation}\label{aeq.03} 
\begin{array}{c} 
\forall (f,j,k):\ \ \ \ s^f_{jk} = r^f_{jk} + o^f_{jk}\beta^f_{jk} \ \ \ \ \mbox{and} \ \
\ \ \forall f \in \CF: \CE^f= \sum_{j,k} o^f_{jk} \beta^f_{jk} 
\end{array}
\end{equation}

\begin{equation}\label{aeq.1} 
\begin{array}{c} 
\forall j \in \CG:  \ \ \  \displaystyle\sum_{i,k} {q_{jk}^i} + \sum_{f,k} r^f_{jk} - \ppj -
\hspace{-0.5cm}\sum_{ f \in \CF(j),j',k}\hspace{-0.5cm} s^f_{j'k} \le \cj \ \ \ \mbox{and}\ \ \  
\ppj (\sum_{i,k} {q_{jk}^i} +
\displaystyle\sum_{f,k} r^f_{jk} - (\ppj+\cj) - \hspace{-0.5cm} \sum_{f\in \CF(j),j',k}\hspace{-0.5cm} s^f_{j'k} ) = 0 
\end{array}
\end{equation}

\begin{equation}\label{aeq.2} 
\begin{array}{c} 
\forall i \in \CA:  \ \ \ \displaystyle\sum_j {w_j^i \ppj} + \sum_f {\theta^i_f \CE^f} - \sum_{j,k}
{q_{jk}^i} \leq - \sum_j {w_j^i}\cj  \ \ \  \mbox{and} 
\ \ \  \li (\displaystyle\sum_{j} {w_j^i (\ppj +
\cj)} + \sum_f {\theta^i_f \CE^f} - \sum_{j,k} {q_{jk}^i} ) = 0 
\end{array}
\end{equation}

\begin{equation}\label{aeq.3} 
\begin{array}{c} 
\forall (i,j,k):  \ \ \ \  u^i_{jk} \li - \ppj - \gijk \leq \cj \ \ \ \  \mbox{and}\ \ \ \
q^i_{jk}(u^i_{jk} \li -  (\ppj + \cj) - \gijk) = 0 
\end{array}
\end{equation}

\begin{equation}\label{aeq.4} 
\begin{array}{c} 
\forall (i,j,k):  \ \ \ \  q^i_{jk} - l^i_{jk} \ppj \leq l^i_{jk}\cj \ \ \ \  \mbox{and}  \ \
\ \  \gijk ( q^i_{jk} - l^i_{jk} (\ppj +\cj)) = 0 
\end{array}
\end{equation}
}

The following lemmas establish a strong connection between solutions of NHAD-LCP and the equilibrium conditions for market $\CM$.

\begin{lemma} \label{alem.optprod} In any solution to NHAD-LCP, each firm produces as per profit maximizing production plan
w.r.t. the prices of goods given by this solution.  \end{lemma}
\begin{proof}
Given a solution, consider prices $\pp=\ppp+\cc$. 
Let $x^f_{jk}=\nfrac{r^f_{jk}}{p_j}$ and
$y^f_{jk}=\nfrac{s^f_{jk}}{p_{j_f}}$ be the amount of used and produced goods on segment $(f,j,k)$.
These are well defined since $p_j$'s are positive. 
Now we will show that $x^f_{jk}$'s, $\beta^f_{jk}$'s and $\pp$ satisfy the optimality conditions (\ref{aeq01}) and (\ref{aeq02}), and
also $y^f_{jk}=\alpha^f_{jk}x^f_{jk}$, for all $(f,j,k)$. 

For a firm $f$ observe that $x^f_{jk}$'s and $\beta^f_{jk}$'s are non-negative and satisfy conditions (\ref{aeq01}) and
(\ref{aeq02}) since (\ref{aeq.01}) and (\ref{aeq.02}) are satisfied at the given solution. Therefore they form a solution of
LP (\ref{alp}) at prices $p_j$ and hence $x^f_{jk}$'s give the amounts to be used at optimal production plan.

Next we show that $y^f_{jk}=\alpha^f_{jk}x^f_{jk}$ indeed holds.
If $x^f_{jk}>0$ then from (\ref{aeq.01}) we have $\beta^f_{jk}=\alpha^f_{jk}p_{j_f}-p_j$. If $\beta^f_{jk}=0$ then from
(\ref{aeq.03}) we have $s^f_{jk}=r^f_{jk}$. In this case,
\[
y^f_{jk} = \frac{s^f_{jk}}{p_{j_f}} = \alpha^f_{jk}\frac{r^f_{jk}}{p_j} = \alpha^f_{jk} x^f_{jk}
\]

If $\beta^f_{jk}>0$ then using (\ref{aeq.02}') we have $r^f_{jk}=o^f_{jk}p_j$ and $x^f_{jk}=o^f_{jk}$, and using (\ref{aeq.03})
$s^f_{jk}=r^f_{jk}+o^f_{jk}\beta^f_{jk}$.
\[
y^f_{jk}=\frac{s^f_{jk}}{p_{j_f}} = \frac{r^f_{jk}+o^f_{jk}\beta^f_{jk}}{p_{j_f}} = \frac{o^f_{jk}p_j + o^f_{jk} (\alpha^f_{jk}p_{j_f}
- p_j)}{p_{j_f}} = \alpha^f_{jk}x^f_{jk}
\]

Further, since $\beta^f_{jk}$ captures profit on segment $(f,j,k)$, if non-negative, variable $\CE^f$ captures the total profit of firm
$f$ (due to (\ref{aeq.03}')).
\end{proof}

\begin{lemma} \label{alem.optbund} 
Each agent receives an optimal bundle of goods at a solution of NHAD-LCP, w.r.t. the prices
of goods given by this solution.  
\end{lemma}
\begin{proof} 
Given a solution let the prices be $\pp=\ppp+\cc$.
Consider an agent $i$. First observe that $\li > 0$, for otherwise (\ref{aeq.3}) will be satisfied as a strict
inequality hence forcing, via (\ref{aeq.3}'), $\qijk = 0$ for each segment of $i$ and hence contradicting market clearing.

Among all segments of $i$ on which a positive allocation has been made, consider one having the lowest bang-per-buck, say it
is $(i, j, k)$.  Let $Q$ be the partition it belongs to and let its bang-per-buck be \[ {\uijk \over p_j} = {1 \over
\sigma_i}  .\] 
\indent Now, by the constraint (\ref{aeq.3}) for this segment, and the non-negativity of $\gijk$, we get that $\li
\geq \sigma_i$.

Define $Q$ to be the flexible partition, all partitions having bang-per-buck strictly higher than $1/\sigma_i$ to be forced
partitions, and all partitions having bang-per-buck strictly lower than $1/\sigma_i$ to be undesirable partitions. Next, we
will prove that the names given are in accordance with those in Section \ref{asec.charNlcp}.

Consider an arbitrary segment of $i$, say $(i, j, k)$. If it is in a forced partition, it must have $\gijk > 0$ in order to
satisfy (\ref{aeq.3}). As a result, in order to satisfy (\ref{aeq.4}'), the inequality (\ref{aeq.4}) must be satisfied with
equality, i.e., this segment is fully allocated. And if $(i, j, k)$ is in an undesirable partition, it must satisfy
(\ref{aeq.3}) as a strict inequality. Hence, $\qijk = 0$ by (\ref{aeq.3}'), i.e., it is totally unallocated.

Finally, if $(i, j, k) \in Q$, there are two cases. If $\li > \sigma_i$, then in order to satisfy (\ref{aeq.3}), $\gijk > 0$.
Again, in order to satisfy (\ref{aeq.4}'), the inequality (\ref{aeq.4}) must be satisfied with equality, i.e., all segments in
partition $Q$ must be fully allocated. And if $\li = \sigma_i$, $\gijk = 0$ in order to satisfy (\ref{aeq.3}). As a result,
the only constraints on $\qijk$ are that $0 \leq \qijk \leq \lijk p_j$, i.e., the allocation on this segment is flexible. In
order to satisfy market clearing, in both cases, the total money spent on segments in $Q$ must exhaust all the money of $i$
that is remaining after all forced partitions are allocated. 

In both cases we get that $1 / \li$ is a lower bound on the bang-per-buck of the flexible partition, i.e., $1 / \sigma_i$,
as was promised. Also, by the characterization given in Section \ref{asec.charNlcp}, $i$ receives an optimal bundle of goods.
\end{proof}

Market clearing, proved in next lemma, is relatively easy.

\begin{lemma}\label{alem.clear}
Given a solution of NHAD-LCP, constraints of (\ref{aeq.1}) and (\ref{aeq.2}) ensures that market clears w.r.t. the prices at the
solution. 
\end{lemma}
\begin{proof} 
Adding the constraints in (\ref{aeq.1}) over all goods and those in (\ref{aeq.2}) over all agents, and using (\ref{aeq.03}) we
get \[ \sum_{i,j,k} {q_{jk}^i} \leq  \sum_j {p_j} + \sum_{f,j,k} o^f_{jk}\beta^f_{jk}  \ \ \ \  \mbox{and} \ \ \ \
\sum_{i,j} w_j^i p_j +\sum_{i,f} \theta^i_{f} \CE^f \leq \sum_{i,j,k} {q_{jk}^i}, \] respectively. Since $\sum_{f,j,k}
o^f_{jk}\beta^f_{jk} = \sum_f \CE^f$, $\sum_i \theta^i_{f}=1$ and $\sum_{i,j} {w_j^i p_j} = \sum_j {p_j}$, both
these inequalities are equalities. Finally, by non-negative, all the constraints in (\ref{aeq.1}) and (\ref{aeq.2}) must hold
with equality, hence proving the lemma.  
\end{proof}

Lemmas \ref{alem.optprod}, \ref{alem.optbund} and \ref{alem.clear} establishes all three requirements of market equilibrium. They together
with Lemma \ref{alem.opt1} gives the following theorem.

\begin{theorem}\label{athm.lcp}
The solutions of NHAD-LCP capture exactly the equilibria of market $\CM$ with SPLC production and SPLC utilities, up to scaling.
\end{theorem}

Theorem \ref{athm.lcp} settles the appropriate subcase of the open problem posed by Eaves (1975) \cite{eaves1} and Garg et
al. \cite{GMSV}, of formulating an LCP to capture equilibria of markets with production.

\section{Algorithm}\label{asec.alg}
From Theorem \ref{athm.lcp} computing an equilibrium of market $\CM$ reduces to solving NHAD-LCP, which
has the same form as the formulation given in (\ref{eq.a}) in Appendix \ref{sec.LCP}; equalities (\ref{aeq.03}) can be removed by 
replacing $s^f_{jk}$'s and $\CE^f$ with corresponding expressions. 
Let $M$ and $\pq$ be the matrix and rhs vector formed by the inequalities of NHAD-LCP, and let $\py$ be the variable vector such that
NHAD-LCP can be written as $M\py \le \pq,\ \py\ge 0,\ \py^T(M\py-\pq)=0$. 

Since the rhs vector $\pq$ does have negative entries, namely in (\ref{aeq.2}), Lemke's algorithm is
applicable (Refer to Appendix \ref{sec.LCP} for detailed description of Lemke's algorithm). 
We will add the $z$ variable only in the constraints and complementarity conditions that have a negative rhs. 
Thus we need to make two changes to NHAD-LCP to obtain the augmented LCP, which we call {\bf NHAD-LCP'}. 
First, we change (\ref{aeq.2}) as follows: 

{\small
\begin{equation} \label{aeq.9} 
\begin{array}{c}
\hspace{-2cm}\forall i \in \CA:  \ \   \displaystyle\sum_j {w_j^i \ppj}+\sum_f \theta^i_f\CE^f - \sum_{j,k} {q_{jk}^i} -z
\le  -  \sum_j {w_j^i c_j}  \ \ \ \  \mbox{and} \\ 
\hspace{5cm}\li (\sum_j {w_j^i
(\ppj+\cj)}+\sum_f\theta^i_f\CE^f - \sum_{j,k} {q_{jk}^i} -z)  = 0 
\end{array}
\end{equation}
}

Second, we impose non-negativity on $z$. 
Clearly, solutions of NHAD-LCP' with $z=0$ are solutions of NHAD-LCP as well, and hence are market equilibria (Theorem \ref{athm.lcp}).
Let the polyhedron of NHAD-LCP' be denoted by $\CP'$.
\medskip

\noindent{\bf Degeneracy.}
The Lemke's scheme assumes that the polyhedron associated with the augmented LCP (NHAD-LCP' in our case) is non-degenerate (Appendix
\ref{sec.LCP}).
It turns out that polyhedron $\CP'$ has an inherent degeneracy at points with $z=0$, so we need to
clarify the non-degeneracy assumption we are making. The degeneracy comes about because of the following fact established in
the proof of Lemma \ref{alem.clear}: adding the constraints in (\ref{aeq.1}) over all goods and those in (\ref{aeq.9}) over all
agents yields two identical equations. 

Henceforth, we will say that the polyhedron $\CP'$ corresponding to NHAD-LCP' is non-degenerate if it has no other degeneracy. 
\medskip

\noindent Let $M'$ be the augmented matrix of NHAD-LCP' and $\py'$ be the corresponding variable vector $(\py,z)$.
Recall from Appendix \ref{sec.LCP} that the set of solutions of NHAD-LCP', called $S$, consists of paths and cycles. 
Our algorithm traverses one such path starting from the {\em primary ray} -- unbounded edge of $S$ where $\py=0$. 
Except for the {\em primary ray} all other unbounded edges in $S$ with $z>0$ are called {\em secondary rays}.
Clearly, $z=\max_i \sum_j w^i_jc_j$ and all other variables zero is a solution vertex of 
NHAD-LCP', call it $\py'_0$; it is also the vertex of the {\em primary ray}.

\begin{table}
\caption{Algorithm for markets with SPLC utilities and SPLC production}
\label{atab.alg}
\centering{
\begin{tabular}{|l|}\hline
Initialization: Let $\py' \leftarrow \py'_0$.\\
{\bf While} $z > 0$ in the current solution $\py'$, {\bf do}\\
\hspace{15pt}Suppose at $\py'$ we have $y'_i=0$ and $(M'\py'-\pq)_i=0$, i.e., $i$ is the double label.\\
\hspace{15pt}{\bf If} $(M'\py'-\pq)_i$ just became $0$ at the current vertex, {\bf then} pivot by relaxing $y'_i=0$.\\
\hspace{15pt}{\bf Else}, pivot by relaxing $(M'\py'-\pq)_i=0$.\\
\hspace{15pt}{\bf If} a new vertex is reached, {\bf then} reinitialize $\py'$ with it.\\
\hspace{15pt}{\bf Else} output `Secondary ray'. {\bf Exit}.\\
{\bf Endwhile}
Output solution $\py'$.\\ \hline
\end{tabular}
}
\end{table}

The algorithm can never cycle or get stuck (no double label found) as discussed in Appendix \ref{sec.LCP}. 
It terminates when either $z$ becomes zero or a {\em secondary ray} is reached. In the former case we obtain
a solution of the original NHAD-LCP and hence a market equilibrium (Theorem \ref{athm.lcp}). 
In the latter case, there is no recourse and the algorithm simply aborts without finding a solution.
Next we show that this case never occurs.

\subsection{No secondary rays in polyhedron $\CP'$}\label{asec.sec}
In this section we show that the polyhedron of NHAD-LCP' does not have {\em secondary rays}. The proof is case by case basis and some
what involved as we need to keep track of all the different variables. 
First we establish a few facts about the points of $S$, crucial for the proof. (solutions of NHAD-LCP').
W.r.t. a solution $(\py,z)$ to NHAD-LCP', define the {\em surplus of agent $i$} to be the difference of her earnings and the amount
of money she spends, i.e., $\sum_j w^i_{j} (\ppj+c_j) +\sum_f  \theta^i_f \CE^f - \sum_{j, k} {\qijk}$. 

\begin{lemma}\label{acl.1} At a solution $(\py,z)$ of NHAD-LCP', surplus of every agent is
non-negative, and is at most $z$. Further, if each good is fully sold then $z=0$.
\end{lemma}
\begin{proof} If $\li = 0$ then for each segment $(i, j, k)$ of $i$, (\ref{aeq.3}) is satisfied with strict inequality.
Hence, by (\ref{aeq.3}'), $\qijk = 0$. Hence $i$ does not spend any money and her surplus equals her earnings which is positive, and 
by (\ref{aeq.2}) it is at most $z$.
If $\li > 0$ then by (\ref{aeq.9}'), $z = \sum_j {w_j^i (c_j+\ppj)} + \sum_f {\theta^i_f \CE^f} - \sum_{j,k} {q_{jk}^i}$, which is the
surplus of $i$; a non-negative quantity.

For the second part, 
since each good $j$ is fully sold, (\ref{aeq.1}) hold with equality for all $j$. 
Adding these over all goods we get $\sum_{i,j,k} {\qijk} = \sum_j (\ppj+c_j) + \sum_f \CE^f$. The l.h.s. and r.h.s are the total money spent
and earned by all the agents respectively. 
Therefore, total surplus is zero. However, since surplus of every agent is non-negative they all have to be zero. Further, note that
the r.h.s is strictly positive, hence at least one $\qijk$ is positive. Due to (\ref{aeq.3}') corresponding $\li$ has to be positive,
and in turn due to (\ref{aeq.2}') $z$ is the surplus of this agent, hence $z=0$.
\end{proof}

The proof is by contradiction. Suppose there is a {\em secondary ray}, say $R$, in $\CP'$. 
Recall that secondary ray is an unbounded edge of set $S$ (solutions of NHAD-LCP') other than the primary ray, with $z>0$.
Let $R$ be incident on the vertex $(\yys,\zs)$, with $\zs>0$, and has the direction vector
$(\yyb, \zb)$.  Then $R = \{ (\yys, \zs) + \dl (\yyb, \zb) \ | \ \forall \dl \geq 0 \}.$ The fact that 
every one of these points is a solution of NHAD-LCP', imposes such heavy constraints that no possibility remains assuming the
sufficiency conditions of {\em no production out of nothing} and {\em strong connectivity}.

All the contradictions uses the following simple fact: $(\yys, \zs) + \dl (\yyb, \zb)$ needs to be a
solution of NHAD-LCP' for unbounded values of $\delta$. Let us start by showing that $\yyb \geq 0$ and $\zb \geq 0$.  If not,
for sufficiently large $\dl$ we will get a point that has a negative coordinate, contradicting a non-negativity
constraint on variables.

The vector $\yy$ consists of six types of variables, i.e., $\yy = (\plambda,\pp',\pq,\pgamma,\pr,\pbeta)$ (variables $\ps$
and $\pE$ are just place holders).  Let $\pppb$ denote the price variables in the direction vector $\yyb$.

\begin{lemma} \label{alem.1} It is not possible that $\pppb > 0$.   \end{lemma} 
\begin{proof} Suppose $\pppb > 0$. Then, at every
point of $R$ with $\dl > 0$, $\pp' > 0$ and therefore by (\ref{aeq.1}') every good is fully sold. Hence, by Lemma
\ref{acl.1}, $z = 0$, Now, we have already established that $\ppps \geq 0$ and $\zb \geq 0$, and by definition of a ray,
$\zs > 0$.  Therefore, at every point of $R$ with $\dl > 0$, $z > 0$ leading to a contradiction.  
\end{proof}

Next we consider the case when $\pppb=0$.

\begin{claim}\label{acl.2} If $\pppb=0$ then $\yyb=0$. \end{claim}
\begin{proof}
If $\pppb = 0$ then the price of each good remains constant on ray $R$. In turn optimal production plans of
firms do not change, and by (\ref{aeq.02}) money spent on raw material cannot increase on all segments except for the last
ones.  On the last ones also they cannot increase using (\ref{aeq.1}) and Claim \ref{acl.noprodcycle}, giving $\rrb=0$ and
$\bbb=0$. Therefore, the total quantity of goods in the market remains the same. Since by (\ref{aeq.1}) no good can be
oversold, $\qqb = 0$. Furthermore, the money earned by agent $i$ through her endowment and profit from firms remains
unchanged throughout $R$. Therefore, the forced, flexible and undesirable partitions of $i$ remain unchanged and hence,
corresponding to each of her undesirable and partially allocated segments, $\gijk = 0$ throughout $R$.

A consequence of strong connectivity is that each agent $i$ must be non-satiated for some good, say $j$. Hence there must be
a segment $(i, j, k)$, with $\uijk > 0$, that is undesirable or partially allocated. Now, in order to satisfy the constraint
(\ref{aeq.3}), $\li$ cannot increase, forcing $\llb = 0$. As a result, for a forced segment $(i, j, k)$, $\gijk$ cannot
increase -- otherwise (\ref{aeq.3}') will force $\qijk = 0$. Putting this together with the assertion about undesirable and
partially allocated segments made above, we get that $\ggb = 0$. Hence, $\yyb = 0$.
\end{proof}

\begin{lemma}\label{alem.pr} If $\yyb = 0$ then $\yys = 0$, i.e., $R$ is the primary ray.  \end{lemma}
\begin{proof} Since the direction vector can not be all zeros we have $\zb > 0$. 
Throughout $R$, for each agent $i$ the money spent and money earned remain unchanged; however,
$z$ increases. Therefore, $\sum_j {w_j^i \ppj} +\sum_f \theta^i_f \CE^f - \sum_{j,k} {q_{jk}^i} -z < -\sum_j c_j {w^i_j}$ at each point of
$R$ except possibly at the vertex of polyhedron $\CP'$. Hence $\li$ has to be zero on the rest of the ray, forcing $\lls = 0$.
Therefore, for each segment, (\ref{aeq.3}) is satisfied as a strict inequality, which forces $\qqs = 0$ by (\ref{aeq.3}'). 

If a good $j$ is not used as a raw material then by (\ref{aeq.1}') its $(\ppps)_j$ is zero.  There are no production cycles
by Claim \ref{acl.noprodcycle}.  If a good is used as a raw material, i.e., $r^f_{jk}>0$ for some $(f,j,k)$, consider a
production chain containing it.  In $\ppps$ the $p'_*$ of the last good in this chain is zero by (\ref{aeq.1}'), because it is
not used as a raw material.  This contradicts $r^f_{jk}>0$ by cascading using (\ref{aeq.01}), (\ref{aeq.01}'), and
(\ref{aeq.02}').  These give $\ppps = 0$ and in turn (\ref{aeq.4}') forces $\ggs = 0$. It together with (\ref{aeq.01}') also
forces $\rrs=0$ and in turn (\ref{aeq.02}') forces $\bbs=0$. Altogether we get $\yys = 0$.  
\end{proof}

Combining Claim \ref{acl.2} with Lemma \ref{alem.pr} we get,

\begin{lemma}\label{alem.2}
It can not be the case that $\pppb=0$ on $R$.
\end{lemma}

\begin{lemma} \label{alem.3} 
Assuming {\em no production out of nothing} and {\em enough demand}, if $\pppb \not > 0$ and $\pppb \neq 0$
then $\CM$ violates {\em strong connectivity}.  
\end{lemma} 
\begin{proof}
Assume that $\pppb \not > 0$ and $\pppb \neq 0$. Let $S \subset G$ be the set of goods for which the vector $\pppb$ is zero
and $\overline{S}$ be the remaining goods; by assumption, both these sets are non-empty.  Let $A_1 \subseteq A$ be the set
of agents who are non-satiated by at least one good in $S$.  Clearly, the prices of goods in $S$ remain constant throughout
$R$ and those of goods in $\overline{S}$ go to infinity.  Hence eventually, the bang-per-buck of all segments corresponding
to goods from $S$ will dominate that of goods from $\overline{S}$. 

Let $F$ be the set of firms producing goods of $S$ and $\overline{F}$ be the remaining firms. Similarly, for any firm in
$\overline{F}$ all the segments corresponding to $S$ will be profitable and will dominate that of goods from $\overline{S}$.
Further, firms in $F$ cannot produce anything using goods from $\overline{S}$ and their production does not change on $R$.
Therefore, their revenue remains constant on $R$.

By (\ref{aeq.1}), each good in $\overline{S}$ is fully sold. Now, since only goods in $S$ can remain unsold and their total
amount in the market is constant, the total surplus of all agents is bounded. Since $z \geq 0$, by Claim \ref{acl.1} each
agent has a non-negative surplus and hence the surplus of each agent is bounded. 

Now, consider an agent $i$ who has a good from $\overline{S}$ in her initial endowment. Since her earnings go to infinity
and her surplus is bounded, she must eventually buy up all segments corresponding to goods in $S$ for which she has positive
utility. Similarly, consider a firm in $\overline{F}$. Since the price of the good it produces go to infinity, all its
non-zero segments corresponding to goods in $S$ will be profitable, and hence will be produced fully. We will use these
observations to derive contradictions based on what $A_1$ consist of. 

Suppose $A_1=A$, then by the observation made above, any agent having a non-zero amount of a good from $\overline{S}$
must eventually demand more than the available amount of some good in $S$, and available amount of every good is bounded by
Claim \ref{acl.noprodcycle}.  contradicting (\ref{aeq.1}).

If $A_1=\emptyset$, then consider an arbitrary agent $i$. For strong connectivity to hold, there must be some agent
$i_1$ or a firm $f \in \overline{F}$ such that $i$ has a good for which $i_1$ or $f$ is non-satiated.  Since $A_1 =
\emptyset$ and all the firms in $\overline{F}$ are satiated for all the goods in $S$, this good is from $\overline{S}$.
Hence each agent has a good from $\overline{S}$ in her initial endowment. Let $j \in S$. Now, by the observation made above,
all agents will eventually buy all segments of $j$ for which they have positive utility. Further, by Claim
\ref{acl.noprodcycle} there is at least one good $j\in S$ not getting produced, contradicting (\ref{aeq.1}), since $\desire(j)
> 1$ (due to {\em enough demand} condition).  

Finally, suppose $\emptyset \subset A_1 \subset A$. An agents of $A_1$ do not own any good from $\overline{S}$, otherwise by
observation made above, demand of some good in $S$ eventually goes to infinity, contradicting (\ref{aeq.1}). Further, for the
same reason firms of $\overline{F}$ are satiated for goods in $S$. This implies that in the graph $G(\CM)$ constructed in Definition
\ref{adef.sc} of Section
\ref{asec.strong} on firms and agents, there is no edge from $A_1$ to $A\setminus A_1$, $A_1$ to $\overline{F}$, and $F$ to
$\overline{F}$, implying that agent nodes in graph $G(\CM)$ are not connected $-$ {\em strong connectivity} not satisfied.
\end{proof}

Putting everything together, Lemmas \ref{alem.1}, \ref{alem.2} and \ref{alem.3} give:

\begin{theorem} \label{athm.sec} 
The polyhedron of NHAD-LCP', corresponding to a market $\CM$ under SPLC production and
SPLC utilities, satisfying strong connectivity, no production out of nothing and enough demand, has no secondary rays.  
\end{theorem}

Theorem \ref{athm.sec} directly yields:

\begin{theorem} \label{athm.alg} If a market $\CM$ with SPLC production and SPLC utilities satisfies strong connectivity, no
positive cycle and enough demand, then $\CM$ admits an equilibrium and the algorithm in Table \ref{atab.alg} terminates with
one.  \end{theorem}

Theorem \ref{athm.alg} settles the appropriate case of the open problem, posed by Eaves (1975) \cite{eaves1}, as described in
the Introduction. Our algorithm also gives a constructive proof of the existence of equilibrium for such markets.  \medskip

\begin{remark}
If we run our algorithm on an arbitrary instance, without sufficiency conditions, then we may end up on a secondary ray,
however that does not imply anything whether equilibrium exists or not. This is expected since checking existence even in 
its restriction to exchange markets with SPLC utilities is NP-complete \cite{VY}, 
and any such implication leads to showing NP=co-NP \cite{Megiddo.1988}. 
\end{remark}

\begin{remark}
Note that the enough demand assumption is used only for the case when $A_1=\emptyset$ of Lemma \ref{alem.3}.
Given a market not satisfying this condition, our algorithm may end up on a secondary ray with the only possibility being
this case. In that case, remove goods of $S$ and firms of $F$ from the market. If firm $f \in \overline{F}$
can produce total $d$ amount from a good $j \in S$, then remove $P^f_{j}$ and change endowments of agents as $w^i_{j_f}
=w^i_{j_f} +\theta^i_f d$. Now, this reduced market is still strongly connected, and has an equilibrium. Again apply the
algorithm to find one. Now to get the equilibrium of the original market set prices of goods in $S$ to zero, and keep
production plans of $F$ as they are on the secondary ray. Distribute goods in $S$ freely to make all its segments forced for
all the agents and to ensure production at full capacity by firms of $\overline{F}$ from these goods. It is easy to see that
such a construction indeed gives a market equilibrium of the original market.  This observation implies that even if the
given market does not satisfy {\em enough demand} condition, then one can find its equilibrium by running the algorithm at most
$n$ times.
\end{remark}

\begin{theorem}\label{athm.ppad} 
Assuming strong connectivity, no positive cycle and enough demand, the problem of computing an equilibrium of a market with
SPLC utilities and SPLC production is in PPAD.  
\end{theorem} 
\begin{proof} 
By Theorem \ref{athm.sec}, the Algorithm must converge to an equilibrium. Now, by Todd's result \cite{todd} on the
orientability of the path followed by a complementary pivot algorithm, we get a proof of membership of the problem in PPAD.
\end{proof}

Recall that the polyhedron $\CP'$ corresponding to NHAD-LCP' has inherent degeneracy as explained in Section \ref{asec.alg}. 
The reason is that at any solution to NHAD-LCP with $z = 0$, the market clearing conditions are
satisfied and the dependence in the constraints established in Lemma \ref{alem.clear} holds.  
We have assumed that there are no other degeneracy except this in the polyhedron corresponding $\CP'$.
Let $v$ be a vertex solution to
NHAD-LCP' with $z=0$. Then it is easy to show that there is exactly one $j \in \CG$ with $\ppj=0$ at $v$.  Relaxing $\ppj=0$
gives an unbounded edge, starting at $v$, at which $z$ remains zero.  Therefore, every point of this edge corresponds to a
market equilibrium in which the prices at $v$ are appropriately scaled.

\begin{theorem} \label{athm.odd} 
If a market $\CM$ with SPLC production and SPLC utilities satisfies strong connectivity, no
positive cycle and enough demand, and its polyhedron $\CP'$ corresponding to NHAD-LCP' is  non-degenerate, then $\CM$ has an
odd number of equilibria,  up to scaling.  
\end{theorem} 

\begin{proof} 
As observed in Appendix \ref{sec.LCP}, the set of
solutions $S$ to NHAD-LCP' consists of paths and cycles.  The solutions of NHAD-LCP' satisfying $z = 0$ are precisely the
solutions to NHAD-LCP. By Theorem \ref{athm.lcp}, the latter exactly capture the equilibria of market
$\CM$, up to scaling.  Now, the solutions of NHAD-LCP' satisfying $z = 0$ occur at endpoints of such paths (under
non-degeneracy).  One of the paths starts with the primary ray and ends with an equilibrium. Since by Theorem \ref{athm.sec}
$\CP'$ has no secondary rays, the rest of the equilibria must be paired up. Hence there are an odd number of equilibria.
\end{proof}

\section{Strongly Polynomial Bound}\label{asec.constant} 
In this section we show that our algorithm is strongly polynomial 
when either the number of goods or number of agents plus number of firms is constant. 
We show a strongly polynomial bound on the number of vertices in the solution set $S$ of
NHAD-LCP', for each of the case. This in turn gives a strongly polynomial bound for our algorithm which traverses a path in $S$. 

In each case we create regions in a constant dimensional space by introducing strongly polynomially many hyperplanes.
We note that the number of non-empty regions formed by $N$ hyperplanes in $\R^d$ is at most $O(N^d)$. Thus we get strongly
polynomial bound on number of regions. After this we show that at most two vertices of $S$ can map to a region thus created.
We extend the construction of \cite{GMSV}. 
The crucial addition in both the cases is to capture the optimal
production plan for each firm and the uncertainty about the amount of good available to agents at equilibrium. 

\subsection{Constant number of goods}

For the constant number of goods consider the cell decomposition in $(p_1,\dots,p_n,z)$-space by adding hyperplanes as
follows: For each segment $(f,j,k)$ of firms add $\alpha^f_{jk}p_{j_f} -p_{j}=0$.  For each $5$-tuple $(i,j,j',k,k')$, where
$i \in \CA$, $j\neq j' \in \CG$, $k\le |U^i_j|$ and $k'\leq |U^i_{j'}|$, introduce hyperplane
${u^i_{jk}}{p_{j'}}-{u^i_{j'k'}}{p_{j}}=0$. These hyperplanes divide the space into cells and each cell has one of the signs
$<,=,>$ for each hyperplane. For firm $f$ let the $Z^f$ contains all the segments $(f,j,k)$ with $\alpha^f_{jk}p_{j_f}
-p_{j}>0$, and let a placeholder $\CE^f = \sum_{(f,j,k) \in Z^f} o^f_{jk} (\alpha^f_{jk}p_{j_f} -p_{j})$.  For each agent,
these signs give partial order on the bang-per-buck of her segments.  Using this information for a given cell, we can sort
all segments $(j,k)$ of agent $i$ by decreasing bang-per-buck, and partition them by equality into classes:
$Q^i_1,Q^i_2,\cdots$.  Let $Q^i_{< l}$ denote $Q^i_1 \cup Q^i_2 \cup \ldots \cup  Q^i_{l - 1}$.  Similarly, we define
$Q^i_{\leq l}$ and $Q^i_{> l}$.

Next we want to capture the flexible partition. To do this, we further subdivide a cell by adding hyperplane $\sum_{(j,k)
\in Q^i_{< l}}$ $l^i_{jk}p_j = \sum_{j\in\CG}w_{ij}p_j + \sum_{f \in \CF} \theta^i_f \CE^f - z$, for each agent $i$ and each
of her partitions $Q^i_l$. For any given subcell, let $Q^i_{l_i}$ be the right most partition such that $\sum_{(j,k) \in
Q^i_{< l_i}} l^i_{jk}p_j < \sum_{j\in\CG}w_{ij}p_j + \sum_{f \in \CF} \theta^i_f \CE^f - z$, then $Q^i_{l_i}$ is the
flexible partition for agent $i$.  In addition, we add hyperplanes $p_j=c_j,\ \forall j \in \CG$ and $z=0$, and consider
only those cells where $p_j\ge c_j$ and $z\ge 0$.

Given a fully-labeled vertex $(\py,z)$ of $\CP'$, there is a natural cell associated with it, namely due to projection of it
on $(\pp,z)$-space by mapping $p'_j$ to $p'_j+c_j$ and $z$ to $z$ itself. 

\begin{lemma}\label{alem.fg} 
At most two vertices of $S$ can map to a region.
Furthermore, if a region is mapped onto from two vertices, then they must be adjacent.  \end{lemma}
\begin{proof} Given a cell we specify one equality for every complementarity condition, to be satisfied by the fully-labeled
vertex mapping to it.  A fully labeled vertex $v=(\plambda,\pp',\pq,\pgamma,\pr,\pbeta,z)$, which maps onto a given cell
must satisfy the following equalities. In the cell,

\begin{itemize} 
\item If $\alpha^f_{jk}p_{j_f} -p_{j} \ge 0$ then $\alpha^f_{jk}p'_{j_f}-\ppj -\beta^f_{jk}=c_{j} -
\alpha^f_{jk}c_{j_f}$ else $r^f_{jk}=0$ at $v$.  
\item If $\alpha^f_{jk}p_{j_f} - p_{j}\le 0$ then $\beta^f_{jk}=0$ else $r^f_{jk}-o^f_{jk}\ppj = o^f_{jk}c_{j}$ at $v$.  
\item If $p_j>c_j$ then $\sum_{i,k} {q_{jk}^i} + \sum_{f,k} r^f_{jk} -
\ppj - \sum_{f\in \CF(j),j',k}s^f_{j'k}  = c_j$ else $\ppj=0$ at $v$.  
\item If $\sum_{j}w^i_jp_j + \sum_{j} \theta^i_f \CE^f -z\ge 0$ (second set of hyperplanes for the tuple $(i,1)$) 
then $\sum_{j} {w_j^i (\ppj + c_j)}+\sum_{f} \theta^i_f \CE^f -
\sum_{j,k} {q_{jk}^i} -z = 0$ else $\lambda_i=0$ at $v$.  
\item If ${u^i_{jk}}{p_{j'}}-{u^i_{j'k'}}{p_{j}}\ge 0$ for a
$(j',k') \in Q^i_{l_i}$ then $u^i_{jk} \li -  \ppj - \gijk = c_j$ else $q^i_{jk}=0$ at $v$.  
\item If ${u^i_{jk}}{p_{j'}}-{u^i_{j'k'}}{p_{j}}\le 0$ for a $(j',k') \in Q^i_{l_i}$ then $\gamma^i_{jk}=0$ else $q^i_{jk} - l^i_{jk}
\ppj = l^i_{jk}c_j$ at $v$.  \end{itemize}

Since the above conditions enforces one equality from each complementary condition of NHAD-LCP', their intersection forms a
line.  If this line does not intersect $\CP'$, no fully labeled vertex gets mapped to the given cell. If it does then
intersection can be either a fully labeled vertex, say $v$, or a fully labeled edge -- we say that an edge of the polyhedron
$\CP'$ is {\em fully labeled} if the solution represented by each point of this edge is fully labeled.  In the former case
only vertex $v$ gets mapped to the cell and in the latter case endpoints of the fully labeled edge map to the cell.
\end{proof}

\subsection{Constant number of agents and firms} 
Let $\CG(\CF)$ be the set of goods produced by the firms, and let they be numbered $\{1,\dots,a\}$ (wlog).
Let $m$ denote the number of agents.
In this section we consider segment configurations for every good. Consider a $\lambda_1,\dots,\lambda_m,p_1,\dots,p_a$-space. 

Again for every segment $(f,j,k)$, where $j\in \CG(\CF)$,
add $\alpha^f_{jk}p_{j_f}-p_{j}=0$. For every segment $(i,j,k)$, where $j \in \CG(\CF)$, add $u^i_{jk}\lambda_i-p_j=0$. These capture
segment configurations for goods in $\CG(\CF)$. Add $\lambda_i \ge 0,\forall i \in \CA$ and $p_j \ge c_j,\forall j \in \CG(\CF)$,
where $c_j$'s are the constants obtained by solving (\ref{aeq.lblp}).

Next for every good outside $\CG(\CF)$ we consider partition of segments $(i,j,k)$ and $(f,j,k)$ together using the following
observation.  Given a fully labeled vertex, for a good $j$ consider the value $u^i_{jk}\lambda_i$ for segments $(i,j,k)$ and
$\alpha^f_{jk}(p'_{j_f}+c_{j_f})$ for segments $(f,j,k)$. Sort them in decreasing order and partition them by equality. It is easy
to verify that at this vertex, good $j$ gets allocated to agents and firms (as a raw material) in the order of partitions,
starting from the first. We call the last allocated partition as flexible partition, all the partitions before it as forced
partitions and all partitions after it as undesirable partitions for good $j$.  Further, let $(i,j,k)$ or $(f,j,k)$ be a
segment in the flexible partition of good $j$, then, we have $u^i_{jk}\lambda_i=p'_j+c_j$ or $\alpha^f_{jk}(p'_{j_f}+c_{j_f})=
p'_j+c_j$ respectively. Therefore, the flexible partition of any good defines its price.

Next we capture the segment configurations for each good not in $\CG(\CF)$. Let $\overline{\CG(\CF)} = \CG\setminus \CG(\CF)$.
Introduce following three types of hyperplanes: 
\begin{itemize} 
\item $u^i_{jk}\lambda_i-u^{i'}_{jk'}\lambda_{i'}=0$ for each
$(i,i',j,k,k')$, where $i\ne i'\in\CA,\ j\in \overline{\CG(\CF)},\ k\le |U^i_{j}|$ and $k'\le |U^{i'}_{j}|$.  
\item $u^i_{jk}\lambda_i-\alpha^f_{jk'}p_{j_f}=0$ for each $(i,j,k,f,k')$, where $i \in \CA$, $j \in \overline{\CG(\CF)}$, $k \le
|U^i_{j}|$, and $k' \le |P^f_j|$.  
\item $\alpha^f_{jk}p_{j_f} - \alpha^{f'}_{jk'}p_{j_{f'}}= 0$ for each
$(f,f',j,k,k')$, where $f\ne f' \in \CF$, $j \in \overline{\CG(\CF)}$, $k \le |P^f_j|$ and $k'\le |P^{f'}_j|$.  
\end{itemize}

Given a cell, the signs of these hyperplanes in the cell give partial order of segments $(i,j,k)$ and $(f,j,k)$ for every
good $j \in \overline{\CG(\CF)}$.  For a good $j$ sort its segments in decreasing order using this partial order, and partition
them by equality in classes: $Q^j_1,Q^j_2,\cdots$.  Since $j$ is not produced its available amount is always $1$.  Next
we capture the flexible partition for good $j$.  For a fully sold good, it may be computed easily by just summing up the
segment lengths, $l^i_{jk}$ for $(i,j,k)$ and $o^f_{jk}$ for $(f,j,k)$, starting from the first partition until it becomes
$1$. However, a fully labeled vertex may have undersold goods. Since the price of such a good is fixed to $c_j$ ($p'_j$ is
zero), segments $(i,j,k)$ and $(f,j,k)$ in its flexible partition have $u^i_{jk}\lambda_i=c_j$ and $\alpha^f_{jk}p_{j_f}=c_{j_f}$.
To capture this we introduce $u^i_{jk}\lambda_i-c_j=0$ for each $(i,j,k)$ and $\alpha^f_{jk}p_{j_f}-c_{j_f}=0$ for each $(f,j,k)$.
In general the flexible partition for good $j$ is the earlier one of the two: partition when good is fully sold and the
partition with value $c_j$. This can be easily deduced for a given cell from the signs of the hyperplanes. Let $Q^j_{l_j}$
be the flexible partition of $j$. 

A fully-labeled vertex $(\py,z)$ maps naturally to $(\plambda,p_1,\dots,p_a)$-space, by mapping $\lambda_i$ to $\lambda_i$
and $p'_j$ to $p'_j+c_j$. Analogues to Lemma \ref{alem.fg} we get the following result. 

\begin{lemma}\label{alem.faf} 
At most two vertices of $S$ can map to a region.
Furthermore, if a region is mapped onto from two vertices, then they must be adjacent.  \end{lemma}

It is clear that our algorithm follows a systematic path instead of a brute force enumeration of all the cells. The next
theorem follows directly from Lemmas \ref{alem.fg} and \ref{alem.faf}, since the number of hyperplane introduced is strongly
polynomial in both the cases.

\begin{theorem}\label{athm_stronglyPoly} For a market, under SPLC production and SPLC utilities, with a constant number of
goods, or agents and firms, our algorithm computes an equilibrium in strongly polynomial time.  \end{theorem}

\section{Reduction: Exchange to Production}\label{asec.exctoprod} 
In this section we give a reduction from an exchange market
to an equivalent Arrow-Debreu market with production and {\em linear utility} functions. The main idea is to introduce a
firm for each agent which essentially produces utility for the
agent. A similar reduction was also given in \cite{JV} but for homogeneous utility functions.

The firms introduced by this reduction output exactly one good. For these firms, we can define the
extreme points of production set $\CY^f$ via a production function $P^f: \Rplus^{n-1} \ra \Rplus$, which takes the amount
of each input good used and outputs the amount of the produced good. 

Consider an exchange market $\CM$ with a set $\CA$ of agents and a set $\CG$ of goods, where utility function of agent $i$
is $U^i$ and $\wi$ is her endowment vector. Let $m=|\CA|$ and $n=|\CG|$. Recall that $U^i$ is a non-decreasing (non-negative),
continuous and concave function.  Now we construct an equivalent Arrow-Debreu market with linear utilities 
$\CM'$: Goods map to goods. For every agent $i$ of $\CM$, we create a good
$n+i$, a firm $i$ with production function $P^i=U^i$ producing good $n+i$, and an agent $i$ with the same initial endowment
and non-zero utility for good $n+i$ only, i.e., for a bundle $\xx=(x_1,\dots,x_{n+m})$, $U'^i(\xx)=x_{n+i}$. Further, agent
$i$ owns the firm $i$, {\em i.e.,} $\theta^i_i = 1$.

Let $\pi^i(\pp)$ denotes the maximum profit of firm $i$ at prices $\pp$. Since among all the firms and agents only agent $i$
wants good $n+i$, all the production of firm $i$ will be bought by her at equilibrium. Further, since agent $i$ is
interested in only good $n+i$ and there is no initial endowment of good $n+i$, she will put all her earnings in buying all
the output of firm $i$. Hence her total budget, $\wi\cdot\pp+\pi^i(\pp)$, should be equal to revenue of firm $i$ at
equilibrium. Let $\pr^i$ be the input bundle used by firm $i$ at an optimal production plan for prices $\pp$.  Equating
these two, we get $\wi\cdot \pp = \pr^i \cdot \pp$.

Using this and the fact that no firm uses the new goods as raw material the next lemma follows easily. For a $C\subseteq \CG$, let
$\pp_C$ denotes the vector with coordinates corresponding to $C$ only.

\begin{lemma}
\label{ale.m'vsm}
There is a one-to-one correspondence between equilibria of $\CM'$ and equilibria of $\CM$.
\end{lemma}
\begin{proof}
Let $\ppp$ be equilibrium prices of $\CM'$ and let $\pp=\ppp_{\CG}$. Let $\pr^i$ be the bundle of raw material used by firm
$i$ at equilibrium. Since $P^i=U^i$ and since $\pr^i$ maximizes profit $p'_{n+i} P^i(\pr)-\pr\cdot\ppp_{\CG}$ and satisfies
$\wi\cdot \ppp = \pr^i \cdot \ppp$, it should also be an optimal bundle of agent $i$ in $\CM$ at prices $\pp$. Market
clearing follows.

For the reverse mapping, it is easy to verify that $\ppp$, where $\ppp_{\CG}=\pp$ and $p'_{n+i}$ be the inverse of the
marginal utility per unit of money of agent $i$, gives an equilibrium for $\CM'$. 
\end{proof}

Examples of ``form $X$'' in the next theorem are linear, SPLC, Leontief, (non-separable) PLC etc.

\begin{theorem}\label{athm.reduction}
An exchange market with utility functions of the form $X$ can be reduced to an equivalent Arrow-Debreu market with linear
utilities and production functions of the form $X$.
\end{theorem}

The next theorem follows using hardness results for exchange markets with SPLC utilities \cite{Chen.plc,ChenTeng,VY,CSVY},
and Theorems \ref{athm.ppad} and \ref{athm.reduction}.

\begin{theorem}\label{athm.hard}
The problem of computing an equilibrium of an AD market with linear utilities and SPLC production is PPAD-complete, 
assuming the weakest known sufficiency conditions by Maxfield \cite{maxfield}.
In general checking existence of an equilibrium in these markets is NP-complete.  
\end{theorem}
\begin{proof}
Containment in PPAD assuming the sufficiency conditions by Maxfield follows from Theorem \ref{athm.ppad}.
Given prices of goods, checking if they are equilibrium prices can be done in polynomial time using the characterization of
Section \ref{asec.charNlcp}, hence containment in NP follows in general.

Both the hardness follows from the hardness results for exchange markets with SPLC utilities
\cite{Chen.plc,ChenTeng,VY,CSVY} and Theorem \ref{athm.reduction}.  
\end{proof}

\section{Experimental Results}\label{asec.exp} 
We coded our algorithm in Matlab and ran it on randomly generated instances of
markets with SPLC production and SPLC utility functions.  Number of segments are kept the same in all the utility and
production functions, let it be denoted by $\#seg$. An instance is created by picking values uniformly at random -- $w^i_j$s
from $[0,\ 1]$, $\theta^i_f$'s from $[0,\ 1]$, $u^i_{jk}$'s from $[0,\ 1]$, $l^i_{jk}$'s from $[0,\ 10/\#seg]$,
$\alpha^f_{jk}$ from $[0,\ 1]$ (in order to avoid positive cycles in production) and $o^f_{jk}$'s from $[0,\ 10/\#seg]$.
For simplicity we assume that firm $a$ produces good $a$.
For every firm $f$, $\theta^i_f$'s are scaled so that they sum up to one. 
Similarly, for every good $j$, $w^i_j$'s are scaled so that they sum up to one. 
For each pair of agent $i$ and good $j$, $u^i_{jk}$'s are
sorted in decreasing order to get PLC $U^i_j$, and similarly for each pair of firm $f$ and good $j$, $\alpha^f_{jk}$'s are sorted in
decreasing order to get PLC $P^f_j$. The experimental results are given in Table \ref{atab.exp}. Note that, even in the worst
case the number of iterations is always linear in the total number of segments in all the input functions. Total number of
segments in a market with $n$ goods, $m$ agents and $l$ firms is $(mn+l(n-1))\#seg$.

\begin{table} \begin{center} \caption{Experimental Results}\label{atab.exp} \vspace{0.2cm} \begin{tabular}{|c|c|c|c|c|} \hline
\#agents, \#goods, \#firms, \#seg & \#Instances & Min & Avg & Max \\ \hline \hline 
2, 2, 2, 2 & 100   &    8 &    13.61& 18\\ \hline 
5, 5, 5, 2 & 100   &   34 &    68.85&    99   \\ \hline 
5, 5, 5, 5 & 100   &  163 &   227.58&   291   \\ \hline
10, 5, 5, 2& 100   &   70 &   119.95&   148   \\ \hline 
10, 5, 5, 5 & 100  &  104 &   396.98&   472   \\ \hline 
10, 10, 10, 2 & 100 &  118 &   224.67&   275  \\ \hline 
10, 10, 10, 5 & 10 &  260 &   714.5 &   905    \\ \hline 
10, 10, 10, 10 & 10 & 326 &  1486.3 &  2210  \\ \hline 
15, 5, 5, 2 & 100  &  141 &   173.24&   214    \\ \hline 
15, 5, 5, 5 & 100  & 500 & 581.42&   684   \\ \hline 
15, 10, 10, 2 & 100 &  219 &   295.84&   374  \\ \hline 
15, 10, 10, 5 & 10  & 1186 &  1934.3 & 2833  \\ \hline 
15, 10, 10, 10 & 10 & 2678 & 2853.2 & 3190   \\ \hline 
\end{tabular}
\end{center} \end{table}

\section{Acknowledgment}

When our first attempts at deriving an LCP for markets with production failed, 
it was Prof. Kenneth Arrow's statement -- if the exchange economy has worked out (in \cite{GMSV}), then so will production -- that kept us going.
We wish to thank him profusely.

\bibliographystyle{abbrv}
\bibliography{kelly}

\appendix
\section{The Linear Complementarity Problem and Lemke's Algorithm}\label{sec.LCP}

Given an $n \times n$ matrix $\MM$, and a vector $\pq$, the linear complementarity problem asks for a vector $\yy$
satisfying the following conditions:

\begin{equation} \label{eq.a} \MM \yy \leq \pq,  \ \ \ \   \yy \geq 0,  \ \ \ \ \pq - \MM \yy \geq 0 \ \ \ \ \mbox{and} \ \
\ \  \yy \cdot (\pq - \MM \yy) = 0.  \end{equation}

The problem is interesting only when $\pq \not \geq 0$, since otherwise $\yy = 0$ is a trivial solution. Let us introduce
slack variables $\pv$ to obtain the equivalent formulation

\begin{equation} \label{eq.b} \MM \yy  + \pv = \pq, \ \ \ \  \yy \geq 0, \ \ \ \ \pv \geq 0 \ \ \ \ \mbox{and} \ \ \ \ \yy
\cdot \pv = 0 .  \end{equation}

Let $\CP$ be the polyhedron in $2n$ dimensional space defined by the first three conditions; we will assume that $\CP$ is
non-degenerate.  Under this condition, any solution to (\ref{eq.b}) will be a vertex of $\CP$, since it must satisfy $2n$
equalities.  Note that the set of solutions may be disconnected.

An ingenious idea of Lemke was to introduce a new variable and consider the system:

\begin{equation} \label{eq.c} \MM \yy  + \pv -z \one  = \pq, \ \ \ \  \yy \geq 0, \ \ \ \ \pv \geq 0, \ \ \ \  z \geq 0  \ \
\ \ \mbox{and} \ \ \ \ \yy \cdot \pv = 0 .  \end{equation}

Let $\CP'$ be the polyhedron in $2n + 1$ dimensional space defined by the first four conditions; again we will assume that
$\CP'$ is non-degenerate.  Since any solution to (\ref{eq.c}) must still satisfy $2n$ equalities, the set of solutions, say
$S$, will be a subset of the one-skeleton of $\CP'$, i.e., it will consist of edges and vertices of $\CP'$.  Any solution to
the original system must satisfy the additional condition $z = 0$ and hence will be a vertex of $\CP'$.

Now $S$ turns out to have some nice properties. Any point of $S$ is {\em fully labeled} in the sense that for each $i$, $y_i
= 0$ or $v_i = 0$.  We will say that a point of $S$ {\em has double label i} if $y_i = 0$ and $v_i = 0$ are both satisfied
at this point. Clearly, such a point will be a vertex of $\CP'$ and it will have only one double label.  Since there are
exactly two ways of relaxing this double label, this vertex must have exactly two edges of $S$ incident at it.  Clearly, a
solution to the original system (i.e., satisfying $z = 0$) will be a vertex of $\CP'$ that does not have a double label.  On
relaxing $z=0$, we get the unique edge of $S$ incident at this vertex.

As a result of these observations, we can conclude that $S$ consists of paths and cycles.  Of these paths, Lemke's algorithm
explores a special one.  An unbounded edge of $S$ such that the vertex of $\CP'$ it is incident on has $z > 0$ is called a
{\em ray}.  Among the rays, one is special -- the one on which $\yy = 0$. This is called the {\em primary ray} and the rest
are called {\em secondary rays}. Now Lemke's algorithm explores, via pivoting, the path starting with the primary ray. This
path must end either in a vertex satisfying $z = 0$, i.e., a solution to the original system, or a secondary ray. In the
latter case, the algorithm is unsuccessful in finding a solution to the original system; in particular, the original system
may not have a solution.  \medskip

\noindent{\bf Remark:}  Observe that $z \one$ can be replaced by $z \pa$, where vector $\pa$ has a 1 in each row in which
$\pq$ is negative and has either a 0 or a 1 in the remaining rows, without changing its role; in our algorithm, we will set
a row of $\pa$ to 1 if and only if the corresponding row of $\pq$ is negative.  As mentioned above, if $\pq$ has no negative
components, (\ref{eq.a}) has the trivial solution $\yy = 0$. Additionally, in this case Lemke's algorithm cannot be used for
finding a non-trivial solution, since it is simply not applicable.

\end{document}